\documentclass[12pt,onecolumn]{IEEEtran}
\usepackage{graphicx}
\usepackage{amsmath}
\usepackage{amssymb}
\usepackage{booktabs}
\usepackage{subfigure}
\usepackage{cases}
\usepackage{float}
\usepackage{scalerel}
\IEEEoverridecommandlockouts

\newenvironment{proof}{\begin{IEEEproof}}{\end{IEEEproof}}

\newcommand{\dE}{{\rm E}} 
\newtheorem{theorem}{Theorem}[section]
\newtheorem{definition}{Definition}[section]
\newtheorem{proposition}{Proposition}[section]
\newtheorem{lemma}{Lemma}[section]

\long\def\symbolfootnote[#1]#2{\begingroup%
\def\thefootnote{\fnsymbol{footnote}}\footnote[#1]{#2}\endgroup}

\def\dref#1{(\ref{#1})}

 \def\sin{\mbox{sin}\,}

\def\be{\begin{equation}} \def\ee{\end{equation}}
\def\ba{\begin{array}} \def\ea{\end{array}} \def\bna{\begin{eqnarray}}
\def\ena{\end{eqnarray}}

\def\dE{\mathbb E}

 \def\bna{\begin{eqnarray}}
\def\ena{\end{eqnarray}} \def\dref#1{(\ref{#1})}

 \allowdisplaybreaks

\IEEEoverridecommandlockouts
\begin{document}
\title{Information Constrained Optimal Transport: \\
From Talagrand, to Marton, to Cover}

\author{Yikun Bai, Xiugang Wu and Ayfer \"{O}zg\"{u}r

\thanks{This paper was presented in part at the 2020 IEEE International Symposium on Information Theory \cite{isit2020}.}

\thanks{Y. Bai and X. Wu  are with the Department of Electrical and Computer Engineering,  University of Delaware, Newark, DE 19716, USA (e-mail: bai@udel.edu; xwu@udel.edu). A. \"{O}zg\"{u}r is with the Department of Electrical Engineering, Stanford University, Stanford, CA 94305, USA (e-mail: aozgur@stanford.edu).}

 
}

\maketitle

\begin{abstract}
The optimal transport problem studies how to  transport one measure to another in the most cost-effective way and has wide range of applications from economics to machine learning. In this paper, we introduce and study an information constrained variation of this problem. Our study yields a strengthening and generalization of Talagrand's celebrated transportation cost inequality. Following Marton's approach, we show that the new transportation cost inequality can be used to recover old and new concentration of measure results. Finally, we provide an application of this new inequality to network information theory. We show that it can be used to recover almost immediately a recent solution to a long-standing open problem posed by Cover regarding the capacity of the relay channel.
\end{abstract}

\section{Introduction}\label{S:Introduction}
The optimal transport (OT) theory, pioneered by Monge \cite{monge1781memoire} and Kantorovich \cite{kantorovich1942translation},  studies how to distribute supply to meet demand in the most cost-effective way. It has many known connections with, and applications to areas such as geometry, quantum mechanics, fluid dynamics, optics, mathematical statistics, and meteorology. More recently, it has received renewed interest due to its increasingly many applications in imaging sciences, computer vision and machine learning.  

\subsection{Optimal Transport  Problem}

The basic OT problem in  Kantorovich's probabilistic formulation can be described as follows. Let $\mathcal Z$ and $\mathcal Y$ be two measurable spaces,  $\mathcal P(\mathcal Z)$ and $\mathcal P(\mathcal Y )$ be the sets of all probability measures on $\mathcal Z$ and $\mathcal Y$ respectively, and $\mathcal P (\mathcal Z \times \mathcal Y)$ be the set of all joint probability measures on $\mathcal Z\times \mathcal Y$. Let $c: \mathcal Z \times \mathcal Y \to \mathbb R_+$ be a non-negative measurable function, which is called the cost function. Given two probability measures $P_Z \in \mathcal P(\mathcal Z)$ and $P_Y \in \mathcal P(\mathcal Y)$, the set of couplings of $P_Z$ and $P_Y$,  denoted by $\Pi(P_Z, P_Y)$, refers to the set of all joint probability measures $P \in \mathcal P (\mathcal Z \times \mathcal Y)$ such that their marginal measures are $P_Z$ and $P_Y$. The OT problem is to find the optimal coupling in $\Pi(P_Z, P_Y)$ that minimizes the expected cost: 
\begin{equation}\label{eq:OT}
\inf_{P\in \Pi(P_Z, P_Y)} \dE_{ P} [c(Z,Y)].
\end{equation}

A special case of particular interest is when $\mathcal Z=\mathcal Y=\mathbb R$ and $c(z,y)= | z-y |^p$, in which case the quantity 
\begin{align}
W_p (P_Z, P_Y) \triangleq \inf_{P\in \Pi(P_Z, P_Y)}\left\{ \dE_{  P} [|Z-Y|^p]\right\}^{1/p}  \label{E:Wd}
\end{align}
defines a distance metric between two probability measures $P_Z$ and $P_Y$ and is called the $p$-th order Wasserstein distance. Various transportation cost inequalities have been developed that upper bound the  Wasserstein distance between two  measures $P_Z$ and $P_Y$. For example, the celebrated Talagrand's transportation inequality \cite{talagrand_transportation_1996} states that 
\begin{align}
W^2_2 (P_{Z}, P_{Y})\leq  2D(P_{Z}\| P_{Y})   \label{E:Talagrand}
\end{align}
when $P_{Y}$ is  standard Gaussian $\mathcal N(0, 1)$ and $P_{Z}\ll P_{Y}$. 

\subsection{Information Constrained Optimal Transport}

In this paper, we propose to study a variation of the OT problem which we call the information constrained OT problem. Here, we want to find the coupling $P$ in $\Pi(P_Z, P_Y)$ that minimizes the expected cost while ensuring that the mutual information $I_P(Z;Y)$ between $Z$ and $Y$ under the coupling $P$ does not exceed some pre-specified value $R$:
\begin{align}
\inf_{P\in \Pi(P_Z, P_Y): I_P (Z;Y)\leq R} \dE_{  P} [c(Z,Y)]. \label{E:equivform}
\end{align}
There are several reasons for us to study this extension of the classical OT problem, which will become clear in the sequel. For now, note that when the infimum in \dref{eq:OT} is achieved by a deterministic mapping between $Z$ and $Y$, the mutual information $I_P (Z;Y)$ will be maximal and can be potentially unbounded. For example, this is known to be the case in \eqref{E:Wd} when $p=2$ and $P_Z$ or $P_Y$ are absolutely continuous with respect to the Lebesgue measure \cite{brenier1987decomposition}. The mutual information constraint in \eqref{E:equivform} can be viewed as enforcing a certain amount of randomization in the mapping between $Z$ and $Y$.  

It is also worth mentioning that an equivalent formulation of the information constrained OT problem has received significant recent interest in the machine learning literature, where one seeks to minimize the cost-information Lagrangian:
\begin{align}
 \inf_{P\in \Pi(P_Z, P_Y)} \left\{ \dE_{ P} [c(Z,Y)] + \lambda I_P (Z;Y) \right\}.\label{E:Lagform}
\end{align}
The problem \eqref{E:Lagform} generally appears under the name entropy regularized OT or Sinkhorn distances. 
In the  machine learning literature, the interest in \eqref{E:Lagform} has been mainly motivated by computational considerations; in many cases computing the regularized OT in \eqref{E:Lagform} from data turns out to be easier than computing the classical OT in \eqref{eq:OT}, which motivates the use of \eqref{E:Lagform} instead of \eqref{eq:OT} as a distance between probability measures \cite{cuturi2013sinkhorn}. For certain inference tasks,  \eqref{E:Lagform} also appears to be a more suitable distance  than  \eqref{eq:OT}, leading to superior empirical performance \cite{courty2016optimal}. Moreover, it is also shown  in \cite{genevay2019sample}--\cite{mena2019statistical} that  \eqref{E:Lagform} can be estimated with much fewer samples as compared to \eqref{eq:OT}.  In contrast to these works which focus on the computational and statistical aspects of \eqref{E:Lagform}, our interest in this paper mainly lies in understanding the solution of the problem \eqref{E:equivform} as well as its fundamental connections to concentration of measure and network information theory. 

\subsection{Summary of Results}
In the information constrained OT setup, one can similarly define the Wasserstein distance between two measures $P_Z$ and $P_Y$ subject to the information constraint $R$:
\begin{align}
W_p (P_Z, P_Y;R)\triangleq \inf_{ \substack{ P\in \Pi(P_Z,P_Y): \\  I_P (Z;Y)\leq R  }}\left\{ \dE_{ P}[ |Z-Y|^p]\right\}^{1/p} .  \label{E:Info-Wd}
\end{align}
Note that when $R=\infty$,  \eqref{E:Info-Wd} reduces to the unconstrained Wasserstein distance in \eqref{E:Wd}. The main result of this paper, proved in Section~\ref{S:NTI}, is an upper bound on $W_2 (P_Z, P_Y; R)$ for any $R\in \mathbb R_+$ when $P_Y$ is standard Gaussian and  $P_Z\ll P_Y$:
\begin{align}
&W^2_2 (P_Z, P_Y; R)   \leq  \dE [Z^2]+1-2\sqrt{\frac{1}{2\pi e} e^{2h(Z)}\left(1-e^{-2R} \right)}  .\label{E:NewIneq}
\end{align} 
This new transportation inequality captures the trade-off between information constraint  and transportation cost, and is tight when $P_Z$ is Gaussian. It can be regarded as a generalization and sharpening of Talagrand's inequality in \eqref{E:Talagrand}. Note that when we take $R\to\infty$  in  \eqref{E:NewIneq}, we get the following bound on the unconstrained Wasserstein distance:
\begin{align}\label{eq:newT2}
W^2_2 (P_Z, P_Y) \leq  \dE [Z^2]+1-2\sqrt{\frac{1}{2\pi e} e^{2h(Z)}   }  . 
\end{align} 
It is easy to check that the R.H.S. of \eqref{eq:newT2} is smaller than or equal to that of Talagrand's inequality in \eqref{E:Talagrand} for any $P_Z$, and therefore \eqref{eq:newT2} is uniformly tighter than \eqref{E:Talagrand}.

Since the pioneering work of Marton \cite{marton1986simple}--\cite{marton1996bounding}, it has been known that Talagrand's transportation inequality captures essentially the same geometric phenomenon as the Gaussian isoperimetric inequality,  both of which can be used to derive concentration of measure in Gaussian space. Do the new transportation inequalities in \eqref{E:NewIneq} and \eqref{eq:newT2} also have natural geometric counterparts? In Section~\ref{S:Geometry}, we show that the strengthening \eqref{eq:newT2} of Talagrand's inequality can be used to prove concentration of measure on the sphere, which can be shown to imply concentration of measure in Gaussian space. In other words, the strengthening  of Talagrand's inequality in \eqref{eq:newT2} captures a stronger isoperimetric phenomenon, the one on the sphere rather than that in Gaussian space. Furthermore, we show in Section~\ref{S:Geometry} that the information constrained transportation inequality in \eqref{E:NewIneq} captures a new isoperimetric phenomenon on the sphere that has not been known before the recent work \cite{barnes2018isoperimetric}--\cite{WuBarnesOzgur_TIT}, co-authored by a subset of the authors. Different from the standard isoperimetric inequality on the sphere where one is interested in the extremal set that minimizes the measure of its neighborhood among all sets of equal measure, this new isoperimetric result deals with the set that has minimal intersection measure with the neighborhood of a randomly chosen point on the sphere. 

Finally, in Section~\ref{S:Application} we demonstrate an application of the information constrained  transportation inequality \eqref{E:NewIneq} to network information theory.  In particular, we show  that it can be used to simplify  the recent solution of a long-standing open problem on communication over the three-node relay channel. Specifically, this problem, ``The Capacity of the Relay Channel'', was posed by Cover in the book 
\emph{Open Problems in Communication and Computation}, Springer-Verlag, 1987 \cite{cover1987capacity}. The recent works \cite{WuBarnesOzgur_TIT,wu2017geometry} solved this problem in the canonical Gaussian case by developing a new converse for the relay channel.\footnote{See also \cite{Leighton1} for the solution in the case of binary symmetric channels.}  The proof in \cite{WuBarnesOzgur_TIT, wu2017geometry} is geometric: the communication problem is recast as a problem about the geometry of typical sets in high-dimensions, and then solved using the new isoperimetric result on the sphere mentioned above. The new transportation inequality \eqref{E:NewIneq} allows us to recover the same result almost immediately, which also enables an interpretation of the previous geometric proof in terms of auxiliary random variables.



\section{New Transportation Inequalities}\label{S:NTI}

Before stating and proving our new transportation inequalities, let us first formalize the definition of the Wasserstein distance and Talagrand's transportation inequality; see also  \cite{raginsky2018concentration}.  Let $(\Omega, d)$ be a Polish metric space. Given $p\ge1$, let $\mathcal P_p(\Omega)$ denote the space of all Borel probability measures $\nu$ on $\Omega$ such that the moment bound
\begin{equation}
\dE_{\omega \sim \nu}[d^p(\omega, \omega_0)]<\infty   
\end{equation}
holds for some (and hence all) $\omega_0 \in \Omega$. 

\begin{definition}[Wasserstein Distance]
Given $p \ge 1$, the  Wasserstein distance of order $p$ between any pair $P_Z,P_Y\in \mathcal P_p(\Omega)$ is defined as
\begin{equation}
W_p (P_Z, P_Y)\triangleq \inf_{ P\in \Pi(P_Z,P_Y)}\left\{ \dE_{  P}[d^p(Z,Y)]\right\}^{1/p}  \label{E:WassDef}
\end{equation} 
where $\Pi(P_Z,P_Y)$ is the set of all probability measures on the product space $\Omega \times \Omega$ with marginals $P_Z$ and $P_Y$.  
\end{definition}
 
Indeed, the function $W_p (P_Z, P_Y)$ of $(P_Z, P_Y)$ in \dref{E:WassDef}  satisfies all the metric axioms  \cite{villani2008optimal} and defines a distance metric on the space $\mathcal P_p(\Omega)$ of distributions. 
%
If $p=2$, $\Omega=\mathbb R$ with $d(z,y)=|z-y|$, and $P_Y$ is atomless, then the optimal coupling that achieves the infimum in \eqref{E:WassDef} is given by the deterministic mapping 
\begin{equation}
    Z=F_{ Z}^{-1}\circ F_{ Y}(Y)  
\end{equation}  
where $F_{Y}$ is the cdf of  $P_Y$, i.e. $F_{Y}(y)=P_Y(Y\leq y)$ and  $F_{Z}^{-1}$ is the quantile function of $P_Z$, i.e. $F_Z^{-1}(\alpha)=\inf\{z\in\mathbb R:F_Z(z)\geq \alpha \}$. 
Building on this optimal coupling and tensorization \cite{raginsky2018concentration}, one can prove the following result for the case when $\Omega=\mathbb R^n$ and  $d(z^n,y^n)= \|z^n-y^n\| _2$, known as Talagrand's transportation inequality.

\begin{proposition}[Talagrand \cite{talagrand_transportation_1996}]\label{P:TI}
 For two probability measures $P_{Z^n}\ll P_{Y^n}$ on $\mathbb R^n$ with $P_{Y^n}$ being  standard Gaussian $\mathcal N(0, I_n)$,  we have 
\begin{equation}
W_2^2(P_{Z^n}, P_{Y^n})\leq 2D(P_{Z^n}\| P_{Y^n}), \label{E:TI}
\end{equation}
where the inequality is tight if and only if $P_{Z^n}$ is a shifted version of $P_{Y^n}$, i.e. $P_{Z^n} = \mathcal N(\mu, I_n)$ for some $\mu \in \mathbb R^n$. 
\end{proposition}


\subsection{Sharpening Talagrand's Transportation Inequality} 

Talagrand's transportation inequality can be sharpened to the following; see also \cite{bakry2012dimension, rioul_almost_2016} for related results.

\begin{theorem} \label{T:NTI}
For $P_{Y^n} = \mathcal N(0, I_n)$  and $P_{Z^n}\ll P_{Y^n}$,  we have 
\begin{equation} 
W_2^2(P_{Z^n}, P_{Y^n})\leq \dE [\|Z^n\|^2]+n -2n \sqrt{\frac{1}{2\pi e}  e^{\frac{2}{n}h(Z^n)}  },
\label{E:NewTI}
\end{equation} 
where the inequality  is tight when $P_{Z^n}$ is isotropic Gaussian, i.e.
$P_{Z^n}=\mathcal N(\mu, \sigma^2 I_n)$ for some $\mu \in \mathbb R^n$ and $\sigma>0$.
\end{theorem}

Note that compared to Talagrand's transportation inequality, which is tight only when $P_{Z^n}=\mathcal N(\mu,  I_n)$, the upper bound of the Wasserstein distance in Theorem \ref{T:NTI} is tight for a wider class of $P_{Z^n}$, i.e. when $P_{Z^n}$ is isotropic Gaussian. If fact, it can been shown (c.f. Appendix \ref{A:Comp}) that this new transportation inequality is in general stronger than Talagrand's, i.e. R.H.S. of  \eqref{E:NewTI} $\leq$  R.H.S. of  \eqref{E:TI},  for any $P_{Z^n}\ll P_{Y^n}$ 
where the inequality holds with equality iff $h(Z^n)=\frac{n}{2}\ln 2\pi e$, which is the case when $P_{Z^n}=\mathcal N(\mu, I_n)$.


%

\subsection{Information Constrained OT} 
 
We next focus on bounding the information constrained OT. 

\begin{definition} [Information Constrained Wasserstein Distance]
Given $p \ge 1$, the Wasserstein distance of order $p$ between any pair $P_Z,P_Y\in \mathcal P_p(\Omega)$ subject to information constraint $R$ is defined as
\begin{equation}
W_p (P_Z, P_Y;R)\triangleq \inf_{ \substack{ P\in \Pi(P_Z,P_Y): \\  I_P (Z;Y)\leq R  }}\left\{ \dE_{ P}[d^p(Z,Y)]\right\}^{1/p} . \label{E:ICWassDef}
\end{equation}  
\end{definition}

It can be verified that the function $W_p (P_Z, P_Y;R)$  of $(P_Z, P_Y )$ in \dref{E:ICWassDef} is nonnegative, symmetric in $(P_Z, P_Y )$, and satisfies the triangle inequality (see \cite{cuturi2013sinkhorn}). However, it is not a true metric despite that we call it information constrained Wasserstein distance, because it violates the coincidence axiom, i.e., $W_p (P_Z, P_Y;R)$ in general is not equal to zero when $P_Z=P_Y$. 
%
For the case when $\Omega=\mathbb R^n$ and  $d(z^n,y^n)= \|z^n-y^n\| _2$, we can prove the following bound on it.
\begin{theorem}\label{T:NTI_Info}
 For $P_{Y^n} = \mathcal N(0, I_n)$  and $P_{Z^n}\ll P_{Y^n}$,  we have 
\begin{align}
 &W^2_2 (P_{Z^n}, P_{Y^n}; R)      
   \leq   \ \dE [\|Z^n\|^2]+n -2n  \sqrt{\frac{1}{2\pi e} e^{\frac{2}{n}h(Z^n)}\left(1-e^{-\frac{2R}{n} } \right)}  .\label{E:NTI_Info}
\end{align}
\end{theorem}

The above theorem characterizes a trade-off between the Wasserstein distance and the information constraint, as depicted in Fig. \ref{F:tradeoff}. This includes Theorem \ref{T:NTI} as an extreme case by letting $R\to \infty$. The other extreme case is when $R=0$, where now $Z^n$ and $Y^n$ are forced to be independent, and therefore the information constrained Wasserstein distance simply reduces to $\dE [\|Z^n\|^2]+n$.

\begin{figure}[htb!]
\centering
\includegraphics[width=0.58 \textwidth]{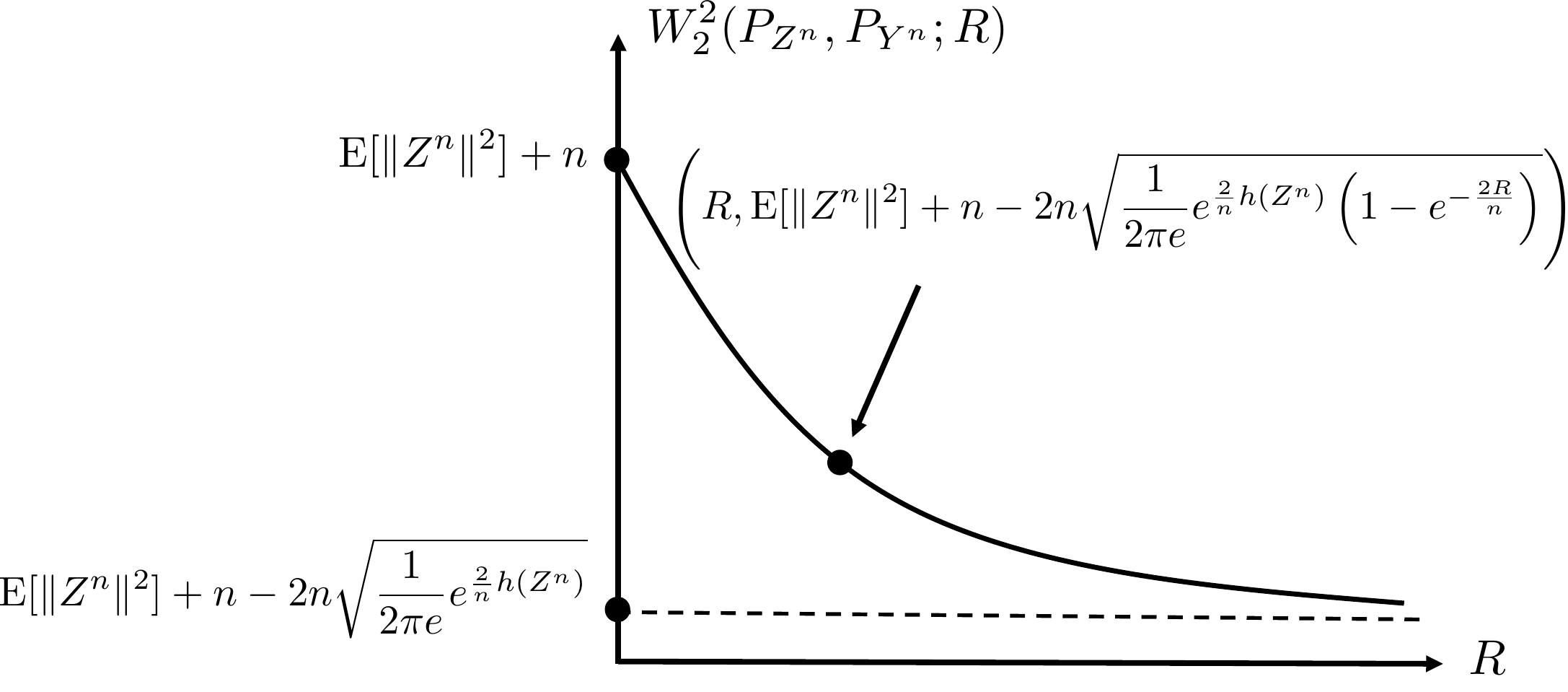}
\vspace{-2mm}
\caption{Wasserstein distance-information constraint tradeoff.}
\label{F:tradeoff}
\end{figure}

The new transportation inequality \eqref{E:NTI_Info} can be shown to be tight when  $P_{Z^n}$ is isotropic Gaussian; that is, when $P_{Z^n}=\mathcal N(\mu, \sigma^2 I_n)$ for some $\mu$ and $\sigma^2$, the inequality in \eqref{E:NTI_Info} is achieved with equality. Therefore, the trade-off characterized in Theorem \ref{T:NTI_Info} is indeed tight when $P_{Z^n}$ is isotropic Gaussian.

 
\subsection{Conditional Transportation Inequality}
 
Both Theorems \ref{T:NTI}  and \ref{T:NTI_Info} have their conditional versions. We start by defining the conditional Wasserstein distance and the conditional information constrained Wasserstein distance.
\begin{definition}[Conditional Wasserstein Distance]
Fix a probability measure $P_T$ and two conditional probability measures $P_{Z|T}$ and $P_{Y|T}$ with
$P_{Z|T=t},P_{Y|T=t}\in \mathcal P_p(\Omega)$ for any $t$. Given $p \ge 1$, the conditional Wasserstein distance of order $p$ between  $P_{Z|T}, P_{Y|T}$ given $P_T$ is defined as
\begin{equation}
W_p (P_{Z|T}, P_{Y|T}| P_T)\triangleq \inf_{ P\in \Pi(P_{Z|T},P_{Y|T}|P_T)}\left\{ \dE_{  P}[d^p(Z,Y)]\right\}^{1/p}  \label{E:CondWassDef}
\end{equation} 
where 
$$\Pi(P_{Z|T},P_{Y|T}|P_T)
\triangleq\{P_{\bar{Z},\bar{Y}|T}\cdot P_T:P_{\bar{Z}|T}=P_{Z|T},P_{\bar{Y}|T}=P_{Y|T}\}.$$ 
\end{definition}

\begin{theorem} \label{T:CNTI}
For any probability measure $P_T$ and  conditional probability measures $P_{Z^n|T}$ and $P_{Y^n|T}$ such that for any $t$, $P_{Y^n|T=t}=P_{Y^n}=\mathcal {N}(0,I_n)$ and $P_{Z^n|T=t}\ll P_{Y^n}$,   we have  
\begin{equation} 
W^2_2 (P_{Z^n|T}, P_{Y^n|T}|P_{T} )  
   \leq  \dE [\|Z^n\|^2]+n -2n  \sqrt{\frac{1}{2\pi e} e^{\frac{2}{n}h(Z^n|T)} }.
\label{E:CNewTI}
\end{equation} 
\end{theorem}

\begin{definition} [Conditional Information Constrained Wasserstein Distance]
Fix a probability measure $P_T$ and two conditional probability measures $P_{Z|T}$ and $P_{Y|T}$ with
$P_{Z|T=t},P_{Y|T=t}\in \mathcal P_p(\Omega)$ for any $t$. Given $p \ge 1$, the conditional Wasserstein distance of order $p$ between  $P_{Z|T}, P_{Y|T}$ given $P_T$ subject to information constraint $R$ is defined as
\begin{align}
W_p (P_{Z|T}, P_{Y|T}|P_T;R)&
\triangleq \inf_{ 
\substack{ P\in \Pi(P_{Z|T},P_{Y|T}|P_T),\\
I_P (Z;Y|T)\leq R  }}\left\{ \dE_{ P}[d^p(Z,Y)]\right\}^{1/p}. \label{E:WassDef_Info_Cond}  
\end{align}  
\end{definition}

 \begin{theorem}\label{T:CNTI_Info}
For any probability measure $P_T$ and  conditional probability measures $P_{Z^n|T}$ and $P_{Y^n|T}$ such that for any $t$, $P_{Y^n|T=t}=P_{Y^n}=\mathcal {N}(0,I_n)$ and $P_{Z^n|T=t}\ll P_{Y^n}$,    we have  
\begin{align}
 W^2_2 (P_{Z^n|T}, P_{Y^n|T}|P_{T}; R) &
   \leq  \dE [\|Z^n\|^2]+n -2n  \sqrt{\frac{1}{2\pi e} e^{\frac{2}{n}h(Z^n|T)}\left(1-e^{-\frac{2R}{n} } \right)}  .\label{E:NTI_Info_Cond}
\end{align}
\end{theorem}

\subsection{Information Density Constrained OT}

We now introduce an OT setup with information density constraint, and present a transportation inequality for this new setup. As we will see, the information density constraint is more stringent than the information constraint, and therefore our previous transportation inequality in Theorem \ref{T:NTI_Info} can be viewed as a special case of this new inequality that we are going to present. This new inequality will be used to prove a new concentration of measure result on the sphere, which has not been known before the recent work \cite{barnes2018isoperimetric}--\cite{WuBarnesOzgur_TIT}; see Proposition \ref{P:strongisoperimetrysphere} and its proof in the next section.

Recall that for a given joint distribution $P\in \cal P(\cal Z \times \cal Y)$  with marginals $P_Z$ and $P_Y$, the information density function $i_P (z;y)$ is defined as  
$$i_P (z;y)=\ln \frac{dP }{d P_{Z} \otimes P_Y }(z,y),$$
whose expectation gives rise to the mutual information $I_P(Z;Y)$, i.e.,
$$I_P(Z;Y)=\dE_P[i_P (Z;Y)].$$
We say that a distribution $P$  satisfies $(R,\tau,\delta)$-information density constraint  for some $R\ge 0$ and $\tau, \delta>0$,  if the following two conditions hold:
\begin{enumerate}
\item the expectation of the information density, i.e. the mutual information, is upper bounded by $R$,
$$I_P(Z;Y)\leq R;$$
\item with probability at least $1-\delta$, the deviation between the information density and mutual information is upper bounded by $\tau$,
\begin{equation*}
\mathbb P_{(Z,Y)\sim P}(|i_P(Z;Y)- I_P(Z;Y)| \leq \tau)\ge 1-\delta  .
\end{equation*}
\end{enumerate}
 
\begin{definition} [Information Density Constrained Wasserstein Distance]
Given $p \ge 1$, the Wasserstein distance of order $p$ between any pair $P_Z,P_Y\in \mathcal P_p(\Omega)$ subject to $(R,\tau,\delta)$-information density constraint is defined as
\begin{equation}
W_p (P_Z, P_Y; R,\tau,\delta )\triangleq \inf_{ \substack{ P\in \Pi(P_Z,P_Y):    I_P (Z;Y)\leq R, \\ \mathbb P_{(Z,Y)\sim P}(|i_P(Z;Y)- I_P(Z;Y)| \leq \tau)\ge 1-\delta   }}\left\{ \dE_{ P}[d^p(Z,Y)]\right\}^{1/p} . \label{E:IDCWassDef}
\end{equation}  
\end{definition}

Compared to the information constrained case, the definition of the information density constrained Wasserstein distance in \dref{E:IDCWassDef} involves an additional constraint $\mathbb P_{(Z,Y)\sim P}(|i_P(Z;Y)- I_P(Z;Y)| \leq \tau)\ge 1-\delta $ in the infimization and therefore given arbitrary $P_Z, P_Y$ and $R$ we have
$$ W_p (P_Z, P_Y; R) \leq W_p (P_Z, P_Y; R,\tau,\delta )  $$
for any $\tau,\delta>0$. 
As in the information constrained case,  the quantity $W_p (P_Z, P_Y; R,\tau,\delta )$  is not a true metric because $W_p (P_Z, P_Y;R,\tau,\delta)$ in general is not equal to zero when $P_Z=P_Y$. However,  it  can be shown to also satisfy a certain triangle inequality under some conditions,  as stated in the following proposition.  The proof of this proposition is included in Appendix \ref{A:TRI}. 

\begin{proposition}\label{P:IDCTRI}
Consider three measures $\mu_1,\mu_2,\mu_3\in \mathcal{P}_p(\Omega)$ such that there exists a one-to-one mapping $g: \Omega \to \Omega$ satisfying that
\begin{enumerate}
\item $\mu_1$ is the push-forward measure of $\mu_2$ under $g$, i.e.,  $\mu_1= \mu_2 \circ g^{-1}$;
\item $g$ induces one optimal coupling that attains $W_p (\mu_1,\mu_2)$, i.e., 
$$W_p (\mu_1,\mu_2)=\left\{ \dE_{ Y\sim \mu_2 }[d^p(g(Y),Y)]\right\}^{1/p}.$$
\end{enumerate}
 Then the following triangle inequality holds: 
\begin{align}
W_p(\mu_1,\mu_3;  R,\tau,\delta )&\leq W_p(\mu_1,\mu_2)+ W_p(\mu_2,\mu_3; R,\tau,\delta ). \label{E:setdelta1} 
\end{align}
\end{proposition}

\vspace{2mm}
For this OT setup with information density constraint, we have the following bound when  $\Omega=\mathbb R^n$ and  $d(z^n,y^n)= \|z^n-y^n\| _2$.

\begin{theorem}\label{P:NTI_InfoD}
For $P_{Y^n} = \mathcal N(0, I_n)$  and $P_{Z^n}\ll P_{Y^n}$, we have that for any $R, \tau \ge 0$
\begin{align}
&W_2^2(P_{Z^n},P_{Y^n};  R, \tau,6n/\tau^2 ) 
\leq      \dE [\|Z^n\|^2]+n -2n  \sqrt{\frac{1}{2\pi e} e^{\frac{2}{n}h(Z^n)}\left(1-e^{-\frac{2R}{n} } \right)} \label{E:NTI_InfoD} 
\end{align}
\end{theorem}

It is easy to see that the above theorem includes Theorem \ref{T:NTI_Info} as a special case by  noting that 
$$W_2^2(P_{Z^n},P_{Y^n};  R ) \leq W_2^2(P_{Z^n},P_{Y^n};  R, \tau ,6n/\tau^2 )$$
for any $R, \tau \ge 0$.

\subsection{Proofs of New Transportation Inequalities}
In this subsection, we provide the proofs of Theorems \ref{T:NTI}--\ref{P:NTI_InfoD}. Recall that Theorems  \ref{T:NTI},  \ref{T:NTI_Info} and  \ref{P:NTI_InfoD} are unconditional  transportation inequalities, while Theorems  \ref{T:CNTI} and \ref{T:CNTI_Info} are the conditional versions. In particular, Theorem  \ref{T:NTI} follows from Theorem \ref{T:NTI_Info}, which  in turn follows from Theorem \ref{P:NTI_InfoD} as a special case. Thus,  in the following we first focus on proving Theorem \ref{P:NTI_InfoD} to establish all the unconditional transportation  inequalities stated in the paper. Then we show how to obtain the conditional versions,  in particular Theorems  \ref{T:CNTI} and \ref{T:CNTI_Info}; for this, it suffices to show how to extend Theorem \ref{T:NTI_Info} 
 to Theorem \ref{T:CNTI_Info}.

\begin{proof}[Proof of Theorem \ref{P:NTI_InfoD}]
To show Theorem \ref{P:NTI_InfoD}, it suffices to construct a coupling $P$ of $P_{Z^n}$ and $P_{Y^n}$ such that the 
$(R, \tau,6n/\tau^2)$-information density constraint is satisfied, i.e., 
 $$ I_P(Z^n;Y^n)   \leq R $$
and
$$\mathbb P_{(Z^n,Y^n)\sim P} (|i_{P}(Z^n;Y^n)- I_P(Z^n;Y^n)|\leq \tau)  \ge 1-6n/\tau^2 ,$$
and simultaneously $\dE_P [\|Z^n-Y^n\|^2]$ is upper bounded by the R.H.S. of \eqref{E:NTI_InfoD}.  For this, let $$Y^n=\sqrt{1-e^{- \frac{2R}{n}}}Y^n_1+e^{-\frac{ R}{n}}Y^n_2,$$ where $Y^n_1,Y^n_2 \sim \mathcal N(0,I_n)$ are independent of each other, and let $Z^n$ satisfy  
$$Z^n=g(Y^n_1)$$ 
for some $g : \mathbb{R}^n \to \mathbb{R}^n$ that pushes $P_{Y_1^n}=\mathcal N(0,I_n)$ forward to $P_{Z^n}$, where $g$ is a differentiable one-to-one mapping whose Jacobian matrix $J_g$ has only nonnegative eigenvalues.
 Note that such a mapping $g$ always exists provided that $P_{Y_1^n}=\mathcal N(0,I_n)$ is absolute continuous with respect to the Lebesgue measure and $P_{Z^n}\ll P_{Y_1^n}$,  and examples include the Brenier mapping \cite{brenier1987decomposition} and the Knothe-Rosenblatt mapping \cite{villani2008optimal}.  See also Lemma 1 of  \cite{rioul_almost_2016}.

It is easy to verify that the joint distribution $P$ of $(Z^n,Y^n)$ defined by the above is   indeed a coupling of $P_{Z^n}$ and $P_{Y^n}$. To see that this coupling satisfies the information density constraint, first note that
\begin{align}
 I_P(Z^n;Y^n)&= h(Y^n)-h(\sqrt{1-e^{- \frac{2R}{n}}}Y^n_1+e^{-\frac{ R}{n}}Y^n_2 |Z^n) \nonumber \\
 &= h(Y^n)-h( e^{-\frac{ R}{n}}Y^n_2|Z^n) \label{e:1-1mapping}  \\
  &= h(Y^n)-h( e^{-\frac{ R}{n}}Y^n_2 ) \label{e:y_2indepz}  \\
& = h(Y^n)-n \ln(e^{-\frac{ R}{n}} ) - h(Y^n_2)  \nonumber  \\
& =R  \nonumber 
\end{align}
where \eqref{e:1-1mapping} holds because $g$ is a one-to-one mapping and thus $Y^n_1$ is determined given $Z^n$, and \eqref{e:y_2indepz} follows from the independence between $Y^n_2$ and $Z^n$. Also we have
\begin{align}
 & \mathbb P (|i_{P}(Z^n;Y^n)- I_P(Z^n;Y^n)|\leq \tau)\nonumber\\
= \ &\mathbb P  \left(\left|\ln\left(\frac{f_{Y^n|Z^n}(Y^n|Z^n)}{f_{Y^n}(Y^n)}\right)-R\right|\leq \tau\right)\nonumber\\
=\ &\mathbb P  \left(\left|\ln\left(\frac{f_{e^{-R/n}Y^n_2}(e^{-R/n}Y^n_2)}{f_{Y^n}(Y^n)}\right)-R\right|\leq \tau\right)\nonumber\\
=\ &\mathbb P  (|\|Y ^n\|^2-\|Y_2^n\|^2|\leq 2\tau) \nonumber \\
\ge\ & \mathbb P( |\|Y ^n\|^2-n| \leq \tau,|\|Y_2^n\|^2-n| \leq \tau)   \nonumber \\
\ge \ & 1-(\mathbb P(|\|Y ^n\|^2-n|\ge \tau)+\mathbb P(|\|Y_2^n\|^2-n|\ge \tau)) \nonumber\\
\ge \ &1-\frac{6n}{\tau^2} \label{Pf:TNI_InfoD_Con2_Ex1}
\end{align}
where \eqref{Pf:TNI_InfoD_Con2_Ex1} holds by Chebyshev's inequality.

Now it remains to show $\dE_P[\|Z^n-Y^n \|^2]\leq$ R.H.S of \eqref{E:NTI_InfoD}. For this, we will lower bound $\dE_P[Z^n\cdot Y^n ]$ in the sequel.  In particular,   letting $g_i$ denote the $i${th} coordinate of $g$, we have 
\begin{align}
  \dE_P [Z^n\cdot Y^n_1 ]
= \ & \sum_{i=1}^n \dE_P \left[\frac{\partial g_i}{\partial y_{1i}}(Y^n_1)\right]\label{Pf:TNI_InfD_Tobound1_Ex1}\\
=\ &\dE_P[\text{trace}(J_g(Y_1^n))]\nonumber\\
\ge\ & \dE_P[n(\det(J_g(Y_1^n)))^{1/n}] \label{Pf:TNI_InfD_Tobound1_Ex2}\\
= \ & n\dE_P[e^{\ln(\det(J_g(Y_1^n)))^{1/n}}] \label{Pf:TNI_InfD_Tobound1_Ex_pos}\\
\ge \ &ne^{\frac{1}{n}\dE_P \ln[\text{det}(J_g(Y^n_1))]}\label{Pf:TNI_InfD_Tobound1_Ex3}\\
=\ & ne^{\frac{1}{n}(h(Z^n)-h(Y^n_1))}\label{Pf:TNI_InfD_Tobound1_Ex4}\\
=\ & n\sqrt{\frac{1}{2\pi e}e^{\frac{2}{n}h(Z^n)}} \nonumber 
\end{align}
where \eqref{Pf:TNI_InfD_Tobound1_Ex1} follows from Stein's lemma for  $P_{Y^n_1}=\mathcal N(0,I_n)$, which says that 
if $Y_1^n\sim \mathcal N(0,I_n)$ and $f:\mathbb R^n\to \mathbb R$ is differentiable, then  
$\dE[f(Y^n_1)Y_{1i}]=\dE[\frac{\partial }{\partial y_i}f(Y^n_1)]$; \eqref{Pf:TNI_InfD_Tobound1_Ex2} holds by the fact that for any  matrix $A$ whose eigenvalues are all nonnegative,  
$\frac{1}{n}\text{trace}(A)\geq (\det(A))^{1/n}$; \dref{Pf:TNI_InfD_Tobound1_Ex_pos} follows from the nonnegativity of $\det(J_g(Y_1^n))$; \eqref{Pf:TNI_InfD_Tobound1_Ex3} is due to Jensen's inequality; and \eqref{Pf:TNI_InfD_Tobound1_Ex4} holds because $Z^n=g(Y_1^n)$ and therefore $$f_{Z^n}(g(y^n_1))\det(J_g(y^n_1) )=f_{Y_1^n}(y^n_1), \forall y_1^n.$$ 
Therefore, $\dE_P[Z^n\cdot Y^n]$ is lower bounded by 
\begin{align}
 \dE_P[Z^n\cdot Y^n] & = \sqrt{1-e^{-\frac{2R}{n}}}\dE_P[Z^n\cdot Y_1^n]\nonumber \\
& \ge  n\sqrt{1-e^{-\frac{2R}{n}}}\sqrt{\frac{1}{2\pi e}e^{\frac{2}{n}h(Z^n)}},\nonumber 
\end{align}
and hence 
\begin{align*}
\dE_P[\|Z^n-Y^n\|^2] = \ & \dE [\|Z^n\|^2]+n -2  \dE_P[Z^n\cdot Y^n] \\
\leq \ &  \dE [\|Z^n\|^2]+n -2n  \sqrt{\frac{1}{2\pi e} e^{\frac{2}{n}h(Z^n)}\left(1-e^{-\frac{2R}{n} } \right)}.
\end{align*}
This completes the proof of Theorem \ref{P:NTI_InfoD}.
\end{proof}

\vspace{2mm}

We now show how to obtain Theorem \ref{T:CNTI_Info} based on Theorem \ref{T:NTI_Info}. 

\begin{proof}[Proof of Theorem \ref{T:CNTI_Info}]
By Theorem \ref{T:NTI_Info},  there exists some $P_{\bar Z^n, \bar Y^n|T }$ such that for any $t$ 
\begin{align}
P_{\bar Y^n|T=t}=\mathcal N(0,I_n), & \ P_{\bar Z^n|T=t}=P_{ Z^n|T=t}, \label{Pf:NTI_InfoCond_Cond1}\\
 \ I(\bar Z^n;\bar Y^n|T=t)& \leq R , \label{Pf:NTI_InfoCond_Cond2}\\
\text{and }  \ \dE[\bar Z^n\cdot \bar Y^n|T=t]& \ge n\sqrt{1-e^{-\frac{2R}{n}}}\sqrt{\frac{1}{2\pi e} e^{\frac{2}{n}h(Z^n|T=t)}}.\label{Pf:NTI_InfoCond_Cond3}
\end{align}
Since \dref{Pf:NTI_InfoCond_Cond1} holds for any $t$, we have $P_{\bar Y^n|T}=\mathcal N(0,I_n)$ and $P_{\bar Z^n|T }=P_{ Z^n|T }$, and hence 
\begin{align}
P_{\bar Z^n, \bar Y^n|T } \in \Pi (P_{  Z^n|T }, P_{  Z^n|T }|P_T).\label{e:conclude_1}
\end{align}
 From \eqref{Pf:NTI_InfoCond_Cond2}, we get 
\begin{align}
I(\bar Z^n; \bar Y^n|T) &= \dE_{S\sim P_T } [I(\bar Z^n; \bar Y^n|T=S) ] \leq \dE_{S\sim P_T } [R ] = R. \label{e:conclude_2}
\end{align} 
Moreover, using \eqref{Pf:NTI_InfoCond_Cond3}  we can lower bound $\dE[\bar Z^n\cdot \bar Y^n]$ by
 \begin{align}
 \dE[\bar Z^n\cdot \bar Y^n]& =\dE_{S\sim P_T}[\dE[\bar Z^n\cdot \bar Y^n|T=S]] \nonumber\\
 & \ge \dE_{S\sim P_T}\left[n \sqrt{1-e^{\frac{2R}{n}}}\sqrt{\frac{1}{2\pi e}e^{\frac{2}{n}h(Z^n|T=S)}}\right] \nonumber\\
& \ge  n\sqrt{1-e^{\frac{2R}{n}}}\sqrt{\frac{1}{2\pi e}e^{\frac{2}{n}\dE_{S\sim P_T}[h(Z^n|T=S)]}} \label{Pf:NTI_InfoCond_bound_ex1}\\
& = n\sqrt{1-e^{\frac{2R}{n}}}\sqrt{\frac{1}{2\pi e}e^{\frac{2}{n}h(Z^n|T)}},\nonumber
 \end{align}
 where \eqref{Pf:NTI_InfoCond_bound_ex1} follows from Jensen's inequality, and therefore $\dE[\|\bar Z^n-\bar Y^n\|^2]$ can be upper bounded by
 \begin{align}
\dE[\|\bar Z^n-\bar Y^n\|^2] = \dE[\|Z^n\|^2]+n -2 n\sqrt{1-e^{\frac{2R}{n}}}\sqrt{\frac{1}{2\pi e}e^{\frac{2}{n}h(Z^n|T)}}   .    \label{e:conclude_3}
 \end{align}
Combining \dref{e:conclude_1}, \dref{e:conclude_2} and \dref{e:conclude_3} completes the proof of Theorem \ref{T:CNTI_Info}. 
\end{proof}

\section{Geometry: Concentration and Isoperimetry}\label{S:Geometry}

Transportation cost inequalities of the form \eqref{E:Talagrand} are known to imply concentration of measure, an inherently geometric phenomenon tightly coupled with isoperimetric inequalites. This section discusses the geometric implications of Theorems \ref{T:NTI}--\ref{T:NTI_Info}. For this, we begin with  the geometry of Talagrand's transportation inequality. 


\subsection{Concentration and Isoperimetry in Gaussian Space}
Consider a Gaussian space $(\mathbb R^n, \gamma)$, where $\gamma=\mathcal N(0,I_n)$ is the standard Gaussian measure on $\mathbb R^n$. For any $A\subseteq  \mathbb R^n$ and $t>0$, let $A_t$ denote the $t$-blowup set of $A$:
$$A_t=\{ x^n\in \mathbb R^n: \|x^n-a^n\|\leq t \text{ for some } a^n \in A  \}.$$
The following concentration of measure result is generally known as the blowing-up lemma in Gaussian space \cite{raginsky2018concentration}.  
\begin{proposition}\label{prop:Gaussblowup}
For any $A\subseteq \mathbb R^n$ with $\gamma_n(A)\geq e^{-na}$, 
\begin{align*}
\gamma_n (A_t) \to 1    \text{ as } n \to \infty 
\end{align*}
when $t \geq \sqrt{2n(a+\epsilon)}$ for some $\epsilon>0$.
\end{proposition}

Roughly, the above result states that under the product Gaussian measure, slightly blowing up any set with a small but exponentially significant probability suffices to increase its probability to nearly 1; hence the name blowing-up lemma. This  lemma can be thought of as a consequence of the isoperimetric inequality in Gaussian space, which says that among all sets with equal Gaussian measure, a halfspace minimizes the measure of its $t$-blowup. Therefore, if we start with two sets $A$ and $H$, where $\gamma(A)=\gamma(H)$ and $H$ is a halfspace, then $\gamma(A_t)\geq \gamma(H_t)$ and hence it suffices to check that $\gamma(H_t)\to 1$, which follows from a simple calculation. 

An alternative approach to proving the above blowing-up lemma, pioneered by Marton \cite{marton1986simple, marton1996bounding}, is through Talagrand's transportation inequality. A formal proof via this approach can be found in \cite{raginsky2018concentration}.  The key observation here is that for any measure $\nu$ and set $A$, $\nu(A)$ can be related to the KL divergence as
$$
D(\nu_A\| \nu)=\ln\frac{1}{\nu(A)},
$$
where $\nu_A$ is the conditional probability measure defined as $\nu_A(C)\triangleq\nu(C\cap A)/\nu(A)$ for any $C$. Together with the triangle inequality for the Wasserstein distance, this allows us to conclude that, for any $A, B\subseteq \mathbb R^n$,
\begin{align*}
W_2 (\gamma_{ A}, \gamma_{ B})  & \leq W_2 (\gamma_{ A}, \gamma ) +W_2 (\gamma_{ B}, \gamma ) \\
& \leq \sqrt{2D(\gamma_A\| \gamma)} + \sqrt{2D(\gamma_B\| \gamma)} \\
& = \sqrt{2\ln\frac{1}{\gamma(A)}} + \sqrt{2\ln\frac{1}{\gamma(B)}}, 
\end{align*}
where the second inequality follows from Talagrand's transportation inequality in \eqref{E:TI}. The proof of Proposition~\ref{prop:Gaussblowup} follows by taking $B=A_t^c\triangleq\mathbb{R}^n\setminus A_t$ and noting that $W_2 (\gamma_{ A}, \gamma_{ A_t^c})\geq t$.

\subsection{Concentration and Isoperimetry on the Sphere}

We next  show that the stronger transportation inequality \eqref{E:NewIneq} also has a natural geometric counterpart. In particular, it implies the following concentration result on the sphere: 
%
Consider a unit sphere $\mathbb{S}^{n-1}\subseteq \mathbb{R}^n$ equipped the uniform probability measure $\mu$ on $\mathbb{S}^{n-1}$,  denoted by $(\mathbb{S}^{n-1}, \mu)$, where $$\mathbb{S}^{n-1}=\left\{ z^n \in\mathbb{R}^n:\|z^n \|=1 \right\}.$$  Recall that a spherical cap with angle $\theta$ on $\mathbb{S}^{n-1}$ is defined as a ball on $\mathbb{S}^{n-1} $ in the geodesic metric (or simply the angle) $\angle(z^n,y^n)=\arccos(\langle z^n, y^n\rangle)$, i.e.,
$$
\text{Cap}(z^n_0,\theta)\triangleq\left\{ z^n\in \mathbb{S}^{n-1}: \angle( z^n_0,z^n)\leq \theta \right\}.
$$
We will say that an arbitrary set $A\subseteq \mathbb{S}^{n-1}$ has an effective angle $\theta$ if $\mu(A)=\mu (C)$, where $C=\text{Cap}(z^n_0,\theta)$ for some arbitrary $ z^n_0\in\mathbb{S}^{n-1}$. In particular, using the the formula for the area of a spherical cap (see \cite[Appendix C]{WuBarnesOzgur_TIT}), we can show that if $A$ has an effective angle $\theta$, then as $n\to\infty$
\begin{equation}
\mu(A)^{1/n}\rightarrow \,\sin\theta.\label{Pf:blowupsphere_muA}
\end{equation}
\begin{proposition}\label{P:blowupsphere}
Let $A \subseteq \mathbb{S}^{n-1}$ be an arbitrary set with effective angle  $\theta\in(0,\pi/2]$. Then for any $\omega>\pi/2 -\theta$, 
\begin{equation}\label{eq:blowup0}
\mu(A_{\omega }) \to 1 \text{ as } n\to \infty,
\end{equation}
where $A_{\omega}$ is the $\omega$-blowup of $A$ defined as 
$$A_{\omega} \triangleq \{x^n \in\mathbb{S}^{n-1}: \angle (z^n, x^n) \leq \omega\;\text{for some}\; z^n \in A  \}.$$
\end{proposition}


As in the case of Gaussian measure concentration, the result in Proposition \ref{P:blowupsphere} is tightly related to the isoperimetric inequality on the sphere. It is easy to see that when $A$ is a spherical cap with angle $\theta$, its blowup $A_{\frac{\pi}{2}-\theta+\epsilon}$ is also a cap (slightly bigger than a halfsphere) whose probability approaches $1$ in high dimensions. Therefore, when $A$ is a spherical cap of angle $\theta$, $\omega=\pi/2-\theta+\epsilon$  is precisely the blowup angle needed for $A_\omega$ to approach probability $1$. Proposition ~\ref{P:blowupsphere} asserts that the same blowup angle is sufficient for any other set $A$ with the same measure, therefore effectively identifying the spherical cap as the extremal set for minimizing the measure of its blowup. 
We next show that Proposition \ref {P:blowupsphere}  can  be derived by properly combining the strengthening \eqref{E:NewTI} of Talagrand's transportation inequality with an argument similar to Marton's procedure.

\begin{proof}[Proof of Proposition \ref{P:blowupsphere}] Fix two sets $A, B \subseteq \mathbb S^{n-1}$ with $\mu (A), \mu(B)>0$. Define the cone extension $\bar A$ of $A$ as
$$\bar A \triangleq \left\{ z^n\in \mathbb R^n:\frac{ z^n}{\|z^n\|}\in A \right\}$$
and define the cone extension $\bar B$ of $B$ similarly. It can be easily seen that the measure of $A, B$ under $\mu$ are the same as the measures of their cone extensions $\bar A, \bar B$ under any rotationally invariant probability measure on $\mathbb R^n$, and in particular, under the standard Gaussian measure $\gamma$, i.e.,
\begin{align*}
\gamma (\bar A )=\mu(A)      \text{ and } \gamma (\bar B )=\mu(B)  .
\end{align*}

Now define two conditional probability measures on $\mathbb R^n$ based on $\bar A, \bar B$:
\begin{align}
\gamma_{A}(C)\triangleq\frac{\gamma (\bar A\cap C)}{\gamma (\bar A)} \text{ and } \gamma_{B}(C)\triangleq\frac{\gamma (\bar B\cap C)}{\gamma (\bar B)} \label{e:defbar}
\end{align}
for arbitrary $C \subseteq \mathbb R^n$. Then $\gamma_{ A}, \gamma_{ B} \ll \gamma$ and we have
\begin{align}
W_2 (\gamma_{ A}, \gamma_{ B}) & \leq W_2 (\gamma_{ A}, \gamma ) +W_2 (\gamma_{ B}, \gamma )  \label{e:san1}\\
&\leq \sqrt{\dE[\|X_A^n\|^2]+n-2n\sqrt{\frac{1}{2\pi e}e^{ \frac{2h(X_{\scaleto{A}{3pt}}^n)}{n} }}} + \sqrt{\dE[\|X_B^n\|^2]+n-2n\sqrt{\frac{1}{2\pi e}e^{ \frac{2h(X_{\scaleto{B}{3pt}}^n)}{n} }}}\label{e:san2}
\end{align}
where $X_A^n \sim \gamma_{ A}$ and $X_B^n \sim \gamma_{ B}$, and \dref{e:san1} follows from the triangle inequality and \dref{e:san2} follows from  Theorem \ref{T:NTI}.  Note that  the density function of $X_A^n$ can be expressed as 
$$
\frac{d\gamma_{ A}}{dx ^n}(x^n) =\frac{\mathbf 1(x^n\in \bar A)}{\gamma  (\bar A)}   \frac{d\gamma }{dx ^n}(x^n) ,
$$
and therefore the second moment $\dE[\|X_A^n\|^2$ is given by
\begin{align}
 \dE_{\gamma_A} [\|X_A^n\|^2] & =\frac{1}{\gamma(\bar A)}\int_{\mathbb R^n} \|x^n\|^2 \mathbf 1(x^n \in \bar A)\gamma (dx^n) \nonumber \\
& =\frac{1}{\gamma(\bar A)}  \dE_{\gamma}[\|X^n\|^2 \mathbf 1(X^n\in \bar A)]\nonumber \\
& = \frac{1}{\gamma(\bar A)}\dE_\gamma [\mathbf 1(X^n\in \bar A)]\dE_\gamma[\|X^n\|^2] \label{Pf:blowupsphere_Expect}\\
 & =n \label{e:asstated1}
\end{align}
and the differential entropy $ h(X_A^n)$ is given by
\begin{align}
 h(X_A^n)&= -\dE_{\gamma_A}\left[\ln \left(  \frac{d\gamma_{ A}}{dx^n}\right)\right] \nonumber \\
 & =-\frac{1}{\gamma(\bar A)}\int_{\mathbb R^n}\ln\left(\frac{1}{\gamma(\bar A)} \frac{d\gamma}{dx^n}(x^n) \right) \mathbf 1 (x^n\in \bar A)  \gamma(dx^n)\nonumber\\
 &=\frac{1}{\gamma(\bar A)}\dE_\gamma\left[\left(\frac{n}{2}\ln(2\pi)+\frac{1}{2}\|X^n\|^2+\ln (\gamma(\bar A))\right) \cdot \mathbf 1(X^n\in \bar A) \right]\nonumber\\
& =\frac{\dE_\gamma[\mathbf 1(X^n\in \bar A)]}{\gamma(\bar A)}\dE_{\gamma}\left[\frac{n}{2}\ln(2\pi)+\frac{1}{2}\|X^n\|^2+\ln (\gamma(\bar A))\right] \label{Pf:blowupsphere_Entropy} \\
 &=\frac{n}{2}  \ln 2\pi e (\gamma  (\bar A))^{2/n} \nonumber  \\
 & = \frac{n}{2}  \ln 2\pi e (\mu  (A))^{2/n}    \label{e:asstated2}
\end{align}
where both \eqref{Pf:blowupsphere_Expect} and \dref{Pf:blowupsphere_Entropy} hold because $\mathbf 1(X^n\in \bar A)$ is independent of $\|X^n\|^2$ (except when $X^n =0$). Similar expressions for $\dE [\|X_B^n\|^2]$ and $ h(X_B^n)$ can also be obtained and thus \dref{e:san2} simplifies to\footnote{Note that applying the original Talagrand's inequality \eqref{E:TI} to $\gamma_{ A}$ and $\gamma_{ B}$ here would yield $W_2 (\gamma_{ A}, \gamma_{ B}) \leq \sqrt{2\ln\frac{1}{\mu( A)}} + \sqrt{2\ln\frac{1}{\mu( B)}}$ instead of \eqref{e:san3}. This inequality is weaker than \eqref{e:san3} and follows from \eqref{e:san3} by using the fact that $\ln x +1\leq x$.}
\begin{align}
W_2 (\gamma_{ A}, \gamma_{ B}) &  \leq \sqrt{2n (1-(\mu(A))^{1/n})    } + \sqrt{2n (1-(\mu(B))^{1/n})    } .   \label{e:san3}
\end{align}
 
On the other hand, we can also obtain a lower bound on $W_2 (\gamma_{ A}, \gamma_{ B}) $. Let $\angle(A,B)$ be the angle distance between $A$ and $B$, defined as 
  $$\angle(A,B)\triangleq \inf\{\angle(x^n,y^n):x^n\in A,y^n \in B\},$$
and assume that $\angle(A,B)\in[0,\pi/2]$ so $\cos(\angle(A,B))\geq 0$.	 To lower bound on $W_2 (\gamma_{ A}, \gamma_{ B}) $, note that for any coupling $P$ of $\gamma_{ A}$ and $\gamma_{B}$ we have
\begin{align}
 \dE_P[\|X^n_A-X^n_B\|^2]&=\dE_{\gamma_A}[\|X^n_A\|]+\dE_{\gamma_B}[\|X ^n_B\|]-2\dE_P[\|X_A^n\|\| X_B^n\|\cos(\angle{(X_A,X_B))}] \nonumber\\
 &\geq 2n-2\dE_P[\|X_A^n\|\| X_B^n\|] \cdot \cos(\angle(A,B))  \nonumber\\
&\ge  2n-2n \cos( \angle(A,B)) \label{Pf:blowupsphere_Dist_ex1} 
\end{align}
where \eqref{Pf:blowupsphere_Dist_ex1} follows from the Cauchy-Schwarz inequality, and therefore we can get the following lower bound on $W_2 (\gamma_{ A}, \gamma_{ B}) $
 \begin{align}
W_2 (\gamma_{ A}, \gamma_{ B}) &  \ge  \sqrt{2n-2n \cos( \angle(A,B))} .   \label{e:san4}
\end{align}
Combining this with \dref{e:san3} gives the following inequality:
 \begin{align}
\sqrt{1-  \cos( \angle(A,B))}  \leq \sqrt{ 1-(\mu(A))^{1/n}   } + \sqrt{ 1-(\mu(B))^{1/n}   }. \label{Pf:blowupsphere_bound}
\end{align}

To finish the proof of Proposition \ref{P:blowupsphere},  fix an arbitrary set $A\subseteq \mathbb S^{n-1}$ with effective angle $\theta\in(0, \frac{\pi}{2}]$ and choose $B=A_\omega^c=\mathbb S^{n-1}\setminus A_\omega$ for $\omega \in (\pi/2-\theta, \pi/2]$. We will use \eqref{Pf:blowupsphere_bound} to show that $\mu(A_\omega^c)\to 0$ as $n\to \infty$. The proof of the proposition for larger $\omega$, follows from the fact that $\mu(A_\omega)$ is increasing in $\omega$. 
Note that by definition, we have  
\begin{align}
\angle(A,A_\omega^c)=\omega. \label{Pf:blowupsphere_DAB}
\end{align}
Plugging this into \dref{Pf:blowupsphere_bound}, and also using \eqref{Pf:blowupsphere_muA} we obtain
 \begin{align}
 \sqrt{ 1-\cos \omega } \leq \sqrt{ 1- \text{sin}\theta}    + \liminf_{n\to\infty} \sqrt{1-(\mu(A_\omega^c))^{1/n} } .
 \end{align}
Therefore, given $\cos \omega < \sin \theta$, i.e.  $\omega>\pi/2 -\theta$,  we have
 \begin{align}
\liminf_{n\to\infty} \sqrt{1-(\mu(A_\omega^c))^{1/n} }>0. 
 \end{align}
This in turn implies that
 \begin{align}
\mu(A_\omega^c) \rightarrow 0   
 \end{align}
as $n\to \infty$, which completes the proof of Proposition \ref{P:blowupsphere}. \end{proof}

\subsection{A New Measure Concentration Result on the Sphere}

We next show that the transportation inequality for information constrained OT leads to a new concentration of measure result on $(\mathbb{S}^{n-1}, \mu)$, which recovers Proposition~\ref{P:blowupsphere} as a special case. This new result was recently proved in \cite{WuBarnesOzgur_TIT, barnes2018isoperimetric} by using Riesz' rearrangement inequality \cite{baernstein} and can be stated as follows:

\begin{proposition}\label{P:strongisoperimetrysphere}
Let $A \subseteq \mathbb{S}^{n-1}$ be an arbitrary set with effective angle  $\theta \in (0,\pi/2]$. Then for any 
$\omega \in (\pi/2 -\theta, \pi/2]$ and $\epsilon>0$,   
\begin{equation}
\mu(  \{y^n: \ln \mu(A\cap \text{Cap}(y^n,\omega ))>   \ln V(\theta, \omega) - n\epsilon \})\to 1,\label{eq:isoperimetry}
\end{equation}
in which $V(\theta, \omega)$ is defined as
\begin{align}
V(\theta, \omega)=\mu(\text{Cap}(z_0^n, \theta) \cap \text{Cap}( y^n_0, \omega)), \label{e: vdef}
\end{align} 
where  $z_0^n, y^n_0$  are perpendicular to each other, i.e.  $\angle(z_0^n, y^n_0)=\pi/2$.  
\end{proposition}

In the proposition, $V(\theta, \omega)$ corresponds to the intersection measure of two spherical caps with  poles perpendicular to each other. By using the surface area formula for the intersection of two spherical caps in \cite[Appendix C-B]{WuBarnesOzgur_TIT}, one can provide an asymptotic characterization of  $\ln V(\theta, \omega)$,
\begin{equation}
\frac{1}{n}\ln V(\theta, \omega)\rightarrow \frac{1}{2}\ln(\sin^2\theta -\cos\omega^2), \text{ as }  n\to\infty.\label{e:plugalpha2}
\end{equation}
Note that an equivalent way to state the blowing-up lemma in Proposition~\ref{P:blowupsphere} is the following: Let $A\subseteq \mathbb{S}^{n-1}$ be an arbitrary set with effective angle $\theta \in (0,\pi/2]$. Then for any $\omega \in (\pi/2 -\theta, \pi/2]$,
\begin{align*} 
\mu (\left\{y^n: \mu(A\cap \text{Cap}(y^n,\omega ))> 0 \right\})\to 1.
\end{align*}
This is true because $\mu(A\cap \text{Cap}(y^n,\omega ))> 0$ if and only if $y^n\in A_{\omega}$. Proposition~\ref{P:strongisoperimetrysphere} extends Proposition~\ref{P:blowupsphere} by providing  a lower bound on  $\mu(A\cap \text{Cap}(y^n,\omega ))$ for $\omega \in (\pi/2 -\theta, \pi/2]$. When $A$ itself is a cap, \eqref{eq:isoperimetry} is straightforward and follows from the fact that $Y^n$ w.h.p. concentrates around the equator at angle $\pi/2$ from the pole of $A$, and therefore the intersection of the two spherical caps is given by $V$ w.h.p.  Proposition~\ref{P:strongisoperimetrysphere} asserts that this intersection measure is  w.h.p. lower bounded by $V$ for any arbitrary $A$ with the same measure. In other words, the spherical cap not only minimizes the measure of its neighborhood as captured by Proposition~\ref{P:blowupsphere}, but roughly speaking, also minimizes its intersection measure with the neighborhood of a randomly chosen point on the sphere.

\begin{proof}[Proof of Proposition \ref{P:strongisoperimetrysphere}]
Fix two sets $A,B\subseteq \mathbb S^{n-1}$ with $\mu (A),\mu (B)>0$. Consider their cone extensions $\bar A, \bar B$ and the induced conditional probability measures $\gamma_A, \gamma_B$ as defined in \eqref{e:defbar}. 
Since $\gamma_B \ll \gamma$ and $\gamma$ is absolutely continuous with respect to the Lebesgue measure,  the optimal coupling that attains  $W_2(\gamma_B,\gamma)$ is a one-to-one mapping that pushes $\gamma$ forward to $\gamma_B$. 
By Proposition \ref{P:IDCTRI}, for any $R\geq 0$ and $\tau>0$ we have 
\begin{align}
&W_2(\gamma_A,\gamma_B;  R,\tau,6n/\tau^2) \label{Pf:StrIsoSph_tobound}\\
\leq\ &  W_2(\gamma_A,\gamma; R,\tau,6n/\tau^2)+W_2(\gamma_B,\gamma) \nonumber \\ 
\leq\ &  \sqrt{\dE[\|X_{\small A}^n\|^2]+n-2n\sqrt{\frac{1}{2\pi e}e^{ \frac{2h(X_{\scaleto{A}{3pt}}^n)}{n} }\left(1-e^{-\frac{2R}{n}}\right)}} + \sqrt{\dE[\|X_B^n\|^2]+n-2n\sqrt{\frac{1}{2\pi e}e^{ \frac{2h(X_{{\scaleto{B}{3pt}}}^n)}{n}}}}
\label{Pf:StrIsoSph_tri_ex1}\\
= \ & \sqrt{2n\left(1-(\mu(A))^{1/n} 
\sqrt{1-e^{-2R/n}}\right)}+\sqrt{2n(1-(\mu(B))^{1/n})}\label{Pf:StrIsoSph_tri} \end{align}
where $X^n_A\sim \gamma_A$ and $X^n_B\sim \gamma_B$;  \eqref{Pf:StrIsoSph_tri_ex1} follows from Theorem \ref{P:NTI_InfoD} and Theorem \ref{T:NTI}; and \dref{Pf:StrIsoSph_tri} follows because $\dE[\|X_A^n\|^2]=n$ and $h(X_A^n)= \frac{n}{2}  \ln 2\pi e (\mu  (A))^{2/n}$ as respectively stated in \dref{e:asstated1} and \dref{e:asstated2}, and  similar expressions hold for  $\dE[\|X_B^n\|^2]$ and $h(X_B^n)$.

On the other hand, we can also obtain a lower bound on 
\eqref{Pf:StrIsoSph_tobound}. 
For any $\eta\in[0,\pi]$, let the function $\alpha(\eta)$ be defined as 
\begin{align}
\alpha(\eta) &\triangleq \frac{1}{n}\left(\ln  \mu (A) -R-\sup_{y^n\in B}\{\ln \mu(\text{Cap}(y^n,\eta)\cap A)\}\right) \label{D:W_InfD_alpha} 
\end{align}
and for any $\epsilon>0$ define the  parameter $\eta_\epsilon^*$ as 
\begin{align}
\eta_\epsilon^*\triangleq \sup\{\eta:\alpha(\eta)\ge \epsilon\}.\label{D:W_InfD_eta*}
\end{align}
The following lemma states a lower bound of \eqref{Pf:StrIsoSph_tobound} in terms of $\eta_\epsilon^*$ that will be useful for proving Proposition  \ref{P:strongisoperimetrysphere}. The proof of this lemma will be presented after we finish the proof of Proposition  \ref{P:strongisoperimetrysphere}.
\vspace{2mm}
\begin{lemma}\label{l:WassInfoDLBound}
For any $\epsilon>0$,
\begin{align*}
W_2(\gamma_A,\gamma_B;R,\tau,6n/\tau^2)\ge \sqrt{2n(1-
\cos \eta_\epsilon^*-\sigma(n, \tau))}
\end{align*}
where 
$\sigma(n, \tau) \to 0$ as $\tau/n,n/\tau^2\to 0$ and $n\to \infty$. 
\end{lemma}
\vspace{2mm}

By lemma \ref{l:WassInfoDLBound} and \eqref{Pf:StrIsoSph_tri}, we get
\begin{align}
\sqrt{1-\cos \eta_\epsilon^*-\sigma(n, \tau)}&
\leq \sqrt{1-(\mu(A))^{1/n} 
\sqrt{1-e^{-2R/n}}}+\sqrt{1-(\mu(B))^{1/n}}, \label{Pf:StrIsoSph_trif}
\end{align}
for any $ \epsilon>0$. To finish the proof of Proposition \ref{P:strongisoperimetrysphere}, fix an arbitrary set $A\subseteq \mathbb S^{n-1}$ with effective angle $\theta\in(0,\frac{\pi}{2}]$ and let 
\begin{align}
B \triangleq \{y^n\in\mathbb S^{n-1}:\ln \mu(A\cap \text{Cap }(y^n,\omega))\leq \ln V(\theta,\omega)-n\beta \}, \nonumber 
\end{align} 
for some arbitrary $\omega\in(\pi/2-\theta,\pi/2]$ and $\beta>0$, where $V(\theta,\omega)$ is as defined in \dref{e: vdef}.  In the sequel, we will  use \dref{Pf:StrIsoSph_trif} to show that $\mu(B) \to 0$ as $n\to \infty$.

To do this, we will apply  \eqref{Pf:StrIsoSph_trif} for a particular choice of $R>0$ and $\epsilon=\frac{\beta}{4}$. Note that \eqref{Pf:blowupsphere_muA} combined with the fact that $\sin\theta>\cos\omega$ implies that
$$
\lim_{n\to\infty} (\mu(A))^{1/n}>\cos\omega.
$$
This implies that there exists a fixed $\phi >0$ such that for sufficiently large $n$
\begin{equation}
(\mu(A))^{1/n}\geq \cos (\omega -\phi).\label{eq:auxeq}
\end{equation}
Therefore,  
letting $R$ be
\begin{align}
    R =  \frac{n}{2} \ln \frac{(\mu(A))^{2/n}}{(\mu(A))^{2/n}-\cos^2(\omega-\phi)} \label{e:defineR}, 
\end{align} 
we have that $R> 0$  for $n$ sufficiently large. We will also assume that $\phi>0$ is chosen sufficiently small so that 
 \begin{align}
    R \leq   \frac{n}{2} \ln \frac{(\mu(A))^{2/n}}{(\mu(A))^{2/n}-\cos^2 \omega } + n\frac{\beta}{8}  \label{e:extra}. 
\end{align}
Note that this is always possible since choosing $\phi$ smaller makes it easier to satisfy \eqref{eq:auxeq}. With the choice of $R$ in \eqref{e:defineR}, the first term on the R.H.S. of \eqref{Pf:StrIsoSph_trif} reduces to  
\begin{align}
\sqrt{1-(\mu(A))^{1/n} 
\sqrt{1-e^{-2R/n}}}=\sqrt{1-\cos (\omega-\phi)}. \label{e:final_comb1}
\end{align}

Now we will focus on the L.H.S. of \eqref{Pf:StrIsoSph_trif} and show that it can be lower bounded by 
$$\sqrt{1-\cos \eta_\epsilon^*-\sigma(n, \tau)} \ge \sqrt{1-\cos \omega -\sigma(n, \tau)} $$
 by choosing $\epsilon=\frac{\beta}{4}$. For this, we evaluate $\alpha(\eta)$ at $\eta=\omega$ under our choice of $A,B$ and $R$:
\begin{align}
\alpha(\omega)&=\ \frac{1}{n}\left(-\sup_{y^n\in B}\{\ln \mu(\text{Cap}(y^n,\omega)\cap A)\}+\ln \mu (A)-R\right) \nonumber \\
&\ge \ \frac{1}{n}(n\beta -\ln V(\theta,\omega)+\ln\mu(A)-R). \label{Pf:StrIsoSph_omega_ex1} 
\end{align}
In \dref{Pf:StrIsoSph_omega_ex1}, we can easily lower bound  $\ln\mu(A)$ by 
\begin{align}
\ln\mu(A) \geq   n \ln \sin \theta -n \frac{\beta}{4}  \label{e:plugalpha1}
\end{align}
for $n$ sufficiently large.  Also, by using \eqref{e:plugalpha2}, we have  
\begin{align}
\ln V(\theta,\omega) & \leq \frac{n}{2}\ln( \sin^2\theta-\cos^2\omega)+n\frac{\beta}{4}  \label{e:plugalpha2}
\end{align}
for $n$ sufficiently large. Moreover, using \eqref{Pf:blowupsphere_muA} in \eqref{e:extra}, $R$ can be further bounded by
 \begin{align}
    R \leq  \frac{n}{2} \ln \frac{\sin^2\theta}{\sin^2\theta-\cos^2 \omega } + n\frac{\beta}{4} \label{e:plugalpha3}, 
\end{align} 
for $n$ sufficiently large.
Plugging \dref{e:plugalpha1}--\dref{e:plugalpha3} into \dref{Pf:StrIsoSph_omega_ex1} , we obtain
\begin{align*}
\alpha(\omega) &\ge    \frac{\beta}{4}, 
\end{align*}
and therefore 
\begin{align}
\omega \leq \eta_{\beta/4}^* \label{Pf:StrIsoSph_omega_ex2}
\end{align}
by the definition of $\eta_{\epsilon}^*$ and the nonincreasing property of $\alpha(\eta)$. Hence, by setting $\epsilon=\beta/4$ on the L.H.S. of \eqref{Pf:StrIsoSph_trif}   and using \dref{Pf:StrIsoSph_omega_ex2}, we obtain
\begin{align}
\sqrt{1-\cos \eta_{\beta/4}^*-\sigma(n, \tau)} \ge \sqrt{1-\cos \omega -\sigma(n, \tau)}. \label{e:final_comb2}
\end{align}
Combining \dref{Pf:StrIsoSph_trif}, \dref{e:final_comb1} and \dref{e:final_comb2} yields 
$$ \sqrt{1-\cos \omega -\sigma(n, \tau)} \leq  \sqrt{1-\cos (\omega-\phi)}+\sqrt{1-(\mu(B))^{1/n}}.  $$
Setting $\tau=n^{3/4}$, we have $\tau/n,n/\tau^2\to 0$ as $n\to \infty$, and thus $\sigma(n, n^{3/4}) \to 0$ as $n\to \infty$.  Therefore, given $\cos \omega <\cos(\omega-\phi)$ and for sufficiently large $n$, we have 
\begin{align}
\mu(B)\leq \left(1-\left(\sqrt{1-\cos\omega-\sigma(n, n^{3/4})}-\sqrt{1-\cos(\omega-\phi)}\right)^2\right)^n,
\end{align}
which tends to zero as $n\to \infty$. 
This completes the proof of Proposition \ref{P:strongisoperimetrysphere}.
\end{proof}

\vspace{2mm}

\begin{proof}[Proof of Lemma \ref{l:WassInfoDLBound}]
Consider an arbitrary coupling $P$ of $(\gamma_A ,\gamma_B )$ that satisfies the $(R,\tau,6n/\tau^2)$-information density constraint.
To find a lower bound on $W_2(\gamma_A,\gamma_B;R,\tau,6n/\tau^2)$, it suffices to lower bound $\dE_P[\|X_A^n-X_B^n\|^2]$,  or equivalently to upper bound $\dE_P[X_A^n\cdot X_B^n]$. Fix $\epsilon>0$ and  define 
$$F=\{\angle(X_A^n,X_B^n)\ge \eta_\epsilon^*,(X^n_A,X_B^n)\in S\}$$
where
 \begin{align}
S\ = &\left\{(x_A^n,x_B^n):|\|x_A^n\|^2-n|\leq \tau ,|\|x_B^n\|^2-n|\leq  \tau ,i_P(x_A^n;x_B^n)\leq R+\tau\right\}.\nonumber
\end{align} 
Then $\dE_P[X_A^n\cdot X_B^n]$ can be upper bounded by conditioning on $F$ and $F^c$ respectively, i.e.,  
 \begin{align*}
\dE_P[X_A^n\cdot X_B^n] & = \dE_P[X_A^n\cdot X_B^n|F ] \mathbb P (F ) +\dE_P[X_A^n\cdot X_B^n|F^c] \mathbb P (F^c)   \\
& \leq  \dE_P[X_A^n\cdot X_B^n|F ]   +\dE_P[X_A^n\cdot X_B^n|F^c] \mathbb P (F^c)    .
\end{align*} 
In the sequel, we will upper bound $\dE_P[X_A^n\cdot X_B^n|F ]$ and $\dE_P[X_A^n\cdot X_B^n|F^c]  \mathbb P (F^c)$ respectively. 

First, from the definition of $F$, we have
 \begin{align}
\dE_P[X_A^n\cdot X_B^n|F ] & = \dE_P[ \|X_A^n\| \|X_B^n\| \cos (\angle(X_A^n,X_B^n))    |F ] \nonumber\\
&\leq (n+\tau) \cos (\eta_\epsilon^*). \label{Pf:StrIsoSph_tobound_bound1} 
\end{align} 
Also, by the Cauchy-Schwarz inequality, we have
\begin{align}
\dE_P[X_A^n\cdot X^n_B |F^c]\mathbb P(F^c)
& \leq \sqrt{\dE_P[\|X^n_A\|^2|F^c]\mathbb P(F^c)}\sqrt{ \dE_P[\|X^n_B\|^2|F^c]\mathbb P(F^c)}\nonumber \\
  &= \sqrt{\dE[\|X_A^n\|^2]-\dE_P[\|X_A^n\|^2|F] \mathbb P(F)}
\sqrt{\dE[\|X_B^n\|^2]-\dE_P[\|X_B^n\|^2|F]\mathbb  P(F)}\nonumber\\ 
 &\leq  n-(n-\tau) \mathbb P(F).\label{Pf:StrIsoSph_tobound_bound2}
\end{align}
To continue with \dref{Pf:StrIsoSph_tobound_bound2}, we need to lower bound $\mathbb P(F)$. Since $\mathbb P(F)$ can be written as
$$\mathbb P(F)=\mathbb P((X_A^n,X_B^n)\in S)-\mathbb P(\angle(X_A^n,X_B^n)\leq \eta_\epsilon^*, (X_A^n,X_B^n)\in S),$$ 
we will bound $\mathbb P((X_A^n,X_B^n)\in S)$ and $\mathbb P(\angle(X_A^n,X_B^n)\leq \eta_\epsilon^*, (X_A^n,X_B^n)\in S)$ respectively. 

To bound $\mathbb P((X_A^n,X_B^n)\in S)$, note that
\begin{align}
\mathbb P(|\|X_A^n\|^2-n|\leq \tau)&=\int_{\bar A} \mathbf 1(|\|x_A^n\|^2-n|\leq r) \gamma_A(dx^n_A) \nonumber \\
&=\frac{1}{\gamma(\bar A)}\int_{\mathbb R^n} \mathbf 1(|\|x^n\|^2-n|\leq \tau)\mathbf 1(x^n\in \bar A)\gamma(dx^n)  \nonumber \\
&=\frac{1}{\gamma(\bar A)}\dE_\gamma[\mathbf 1(|\|X^n\|^2-n|\leq \tau)\mathbf 1(X^n\in \bar A)] \nonumber \\
&=\frac{1}{\gamma(\bar A)}\dE_\gamma[\mathbf 1(|\|X^n\|^2-n|\leq \tau)]\dE_\gamma[\mathbf 1(X^n\in \bar A)] \label{Pf:StrIsoSph_tobound_prob1_ex1}\\
&=\mathbb P (|\|X^n\|^2-n|\leq \tau)\nonumber\\
&\ge 1- 3n/\tau^2, \label{Pf:StrIsoSph_tobound_prob1_ex2} 
\end{align}
where $X^n\sim \gamma$, \eqref{Pf:StrIsoSph_tobound_prob1_ex1} holds because 
$\mathbf 1(|\|X^n\|^2-n|\leq \tau)$ and $\mathbf 1(X^n\in \bar A)$ are independent (except when $X^n=0$), and \eqref{Pf:StrIsoSph_tobound_prob1_ex2} follows from Chebyshev's inequality.  
Similarly, $\mathbb P(|\|X_B^n\|^2-n|\leq \tau)\ge 1-3n/\tau^2$. In addition, since $P$ satisfies the $(R,\tau,6\tau^2/n)$- information density constraint, we have $$\mathbb P(i_P(X_A^n;X_B^n)\leq R+\tau)\ge 1-6n/\tau^2.$$ Therefore, by the union bound we have 
\begin{align}
\mathbb P((X_A^n,X_B^n)\in S)\ge 1-12 n /\tau^2. \label{Pf:StrIsoSph_tobound_prob1}
\end{align}

To upper bound $\mathbb P(\angle(X_A^n,X_B^n)\leq \eta_\epsilon^*, (X_A^n,X_B^n)\in S)$, we have 
 \begin{align}
&\mathbb P(\angle(X_A^n,X_B^n)\leq \eta_\epsilon^*, (X_A^n,X_B^n)\in S) \nonumber\\
& =\ \int_{\bar B} \int_{\bar A}  f_{X_A^n|X_B^n}(x_A^n|x_B^n) \mathbf 1\left( (x_A^n,x_B^n) \in S,\angle(x_A^n,x_B^n)\leq \eta_\epsilon^*\right)  dx_A^nf_{X_B^n}(x_B^n) dx_B^n \nonumber\\
 & \leq\ \int_{\bar B}\int_{\bar {A}}e^{R-h(\gamma_A)+\frac{3}{2}\tau} \mathbf 1((x_A^n,x_B^n)\in S, \angle(x_A^n,x_B^n)\leq \eta_\epsilon^*)dx_A^nf_{X_B^n}(x_B^n)dx_B^n \label{Pf:StrIsoSph_tobound_prob2_ex1}\\ 
 & \leq\ \int_{\bar B}e^{R-h(\gamma_A)+\frac{3}{2}\tau}e^{-n\epsilon+h(\gamma_A)-R+\frac{1}{2} \tau+n\epsilon_1}f_{X^n_B}(x_B^n)dx_B^n\label{Pf:StrIsoSph_tobound_prob2_ex2}\\
 &= \ e^{-n(\epsilon -\frac{2\tau}{n}-\epsilon_1)}\nonumber\\ 
 &\leq \ \epsilon_2,\label{Pf:StrIsoSph_tobound_prob2}
 \end{align}
where $\epsilon_1 \to 0$ as $n\to \infty$, and
$\epsilon_2 \to 0$ as $n\to \infty$ and $\frac{\tau}{n} \to 0$. In the above, \eqref{Pf:StrIsoSph_tobound_prob2_ex1} holds because for each $(x_A^n,x_B^n)\in S\cap (\bar A\times \bar B)$,  the conditional density $f_{X^n_A|X^n_B}(x^n_A|x^n_B)$ satisfies
\begin{align}
f_{X^n_A|X^n_B}(x^n_A|x^n_B)\ &= \ e^{i_P(x^n_A;x^n_B)}f_{X_A^n}(x^n_A)\nonumber \\
\ & \leq \ e^{R+\tau} e^{-\frac{n}{2}\ln (2\pi e \mu(A)^{2/n})+\frac{1}{2}(n-\|x_A^n\|^2)} \label{Pf:StrIsoSph_tobound_prob2_ex1_ex1}
\\& \leq \ e^{R-h(\gamma_A)+\frac{3}{2}\tau} , \label{Pf:StrIsoSph_tobound_prob2_ex1_ex2}
\end{align}
where \eqref{Pf:StrIsoSph_tobound_prob2_ex1_ex1} and \eqref{Pf:StrIsoSph_tobound_prob2_ex1_ex2} follows from the facts that $i_P(x_A^n;x_B^n)\leq R+\tau$ and $|\|x_A^n\|^2-n|\leq \tau$ respectively by the definition of $S$. Inequality
\eqref{Pf:StrIsoSph_tobound_prob2_ex2} holds because for each $x_B^n\in \bar B$, 
we have 
\begin{align}
 \int_{\bar A} \mathbf 1((x_A^n,x_B^n)\in S,\angle(x_A^n,x_B^n) \leq \eta_\epsilon^n)dx_A^n 
&\leq \mu(A\cap \text{Cap}(x_B^n,\eta_\epsilon^*))|B(0,\sqrt{n+\tau})|\nonumber\\
 &\leq    \mu (A\cap \text{Cap}(x_B^n,\eta_\epsilon^*)) e^{\frac{n}{2}\ln (2\pi e )+\frac{1}{2}\tau+n\epsilon_1 } \label{Pf:StrIsoSph_tobound_prob2_ex2_2} \\
& \leq  e^{-n\alpha(\eta_\epsilon^*)+\ln(\mu(A))-R}  e^{\frac{n}{2}\ln (2\pi e )+\frac{1}{2} \tau+n\epsilon_1 } \label{Pf:StrIsoSph_tobound_prob2_ex2_3}\\
& \leq  e^{-n\epsilon+\ln(\mu(A))-R}  e^{\frac{n}{2}\ln (2\pi e )+\frac{1}{2} \tau+n\epsilon_1 } \label{Pf:StrIsoSph_tobound_prob2_ex2_new}\\
&=    e^{-n\epsilon+h(\gamma_A)-R+\frac{1}{2} \tau+n\epsilon_1 }, \nonumber
\end{align}
where $|B(0,\sqrt{n+\tau})=\{x^n:\|x^n\|\leq \sqrt{n+\tau}\}|$ denotes the volume of the Euclidean ball with center $0$ and radius $\sqrt{n+\tau}$. Here, \eqref{Pf:StrIsoSph_tobound_prob2_ex2_2} holds because from \cite[Lemma 13]{WuBarnesOzgur_TIT}, we have 
 $$|B(0,\sqrt{n+\tau})|\leq e^{\frac{n}{2}\ln(2\pi e(1+\frac{\tau}{n}))+n\epsilon_1}\leq e^{\frac{n}{2}\ln 2\pi e+\frac{1}{2}\tau+n\epsilon_1 },
 $$
 where the last inequality uses the fact $\ln(1+a)\leq a$ for any $a\ge 0$,
 \eqref{Pf:StrIsoSph_tobound_prob2_ex2_3} follows from the definition of  $\alpha(\eta)$, and \dref{Pf:StrIsoSph_tobound_prob2_ex2_new} holds  because $\alpha(\eta)$ is continuous in $\eta$ by Lemma \ref{l:W2alphaisCont} and hence $\alpha(\eta_\epsilon^*)\ge \epsilon$.

Combining \eqref{Pf:StrIsoSph_tobound_prob1} and  \eqref{Pf:StrIsoSph_tobound_prob2}, we have
\begin{align}
 \mathbb P (F)&\ge  1 -12n/\tau^2 -\epsilon_2 \nonumber \\
& \ge 1-\epsilon_3 \label{Pf:StrIsoSph_tobound_prob} 
\end{align}
where $\epsilon_3\to 0$ as $n\to \infty$, $n/\tau^2\to 0$ and $\tau/n\to 0$. Combining \eqref{Pf:StrIsoSph_tobound_bound1}, \eqref{Pf:StrIsoSph_tobound_bound2} and \eqref{Pf:StrIsoSph_tobound_prob}, we have 
\begin{align}
\dE_P[X_A^n\cdot X_B^n] 
 \leq n(\cos\eta_\epsilon^* +\sigma(n,\tau))
\end{align}
where $\sigma(n,\tau) \to 0$ as $n\to \infty$, $n/\tau^2\to 0$ and $\tau/n\to 0$, and therefore
$$\dE_P[\|X_A^n-X_B^n\|^2] \ge 2n(1-\cos\eta_\epsilon^*-\sigma(n,\tau)).$$
Since the above inequality holds for any coupling $P$ of $(\gamma_A,\gamma_B)$ that satisfies the $(R,\tau,6n/\tau^2)$-information constraint, we can conclude that 
$$W_2(\gamma_A,\gamma_B;R,\tau,6n/\tau^2)\geq 2n(1-\cos\eta_\epsilon^*-\sigma(n,\tau)).$$
This completes the proof of Lemma \ref{l:WassInfoDLBound}.
\end{proof}

\section{An Application to Network Information Theory}\label{S:Application}

We next demonstrate an application of our transportation inequalities in network information theory. In particular, we show that the information constrained transportation  inequality  can be used to recover the recent solution of a  problem posed by Cover in 1987 \cite{cover1987capacity} regarding the capacity of the relay channel.

To describe Cover's problem, consider a Gaussian primitive relay channel given by 
\begin{numcases}{}
Z=X+W_1\nonumber \\
Y=X+W_2 \nonumber
\end{numcases}
where $X$ denotes the source signal constrained to average power $P$, $Z$ and $Y$ denote the received signals of the relay and the destination  respectively, and $W_1\sim \mathcal N(0,N)$ and $W_2\sim \mathcal N(0,1)$ are Gaussian noises that are independent of each other and $X$. The relay channel is  ``primitive'' in the sense that the relay is connected to the destination with an isolated bit pipe of capacity $C_0$.   Let $C(C_0)$ denote the capacity of this relay channel as a function of $C_0$. What is the critical value of $C_0$ such that $C(C_0)$ first equals $C(\infty)$? This is problem posed by Cover in \emph{Open Problems in Communication and Computation}, Springer-Verlag, 1987 \cite{cover1987capacity},  which he calls ``The Capacity of the Relay Channel''. 

This question was answered in a recent work \cite{WuBarnesOzgur_TIT, wu2017geometry}, which shows that $C(C_0)$ can not be equal to $C(\infty)$ unless $C_0=\infty$, regardless of the SNR of the Gaussian channels. This result follows as a corollary to a new upper bound developed in \cite{WuBarnesOzgur_TIT, wu2017geometry} on the capacity of this channel, which builds on a strong data processing inequality (SDPI) for a specific Markov chain. The proof of this SDPI in \cite{WuBarnesOzgur_TIT, wu2017geometry} is geometric and heavily relies on the new measure concentration result stated in Proposition \ref{P:strongisoperimetrysphere}. We next show that the transportation inequality we develop in the current paper can also be used to establish this SDPI providing a much shorter and simpler proof. We now state the SDPI and briefly illustrate how it leads to a new upper bound on the relay channel. We then prove it by using the conditional version of the information constrained transportation inequality as stated in Theorem \ref{T:CNTI_Info}.

\subsection{A Strong Data Processing Inequality}

Consider a long Markov chain 
\begin{align}Y^n-X^n-Z^n-U_n, \label{e:lmc}\end{align} 
with $Z^n=X^n+W_1^n$ and $Y^n=X^n+W_2^n$, where $\dE[\|X^n\|^2]=nP$, $W_1^n \sim \mathcal N(0, N I_n)$, $W_2^n\sim \mathcal N(0, I_n)$, and $X^n,W_1^n,W_2^n$ are mutually independent. For this long Markov chain, the following SDPI was established in \cite{WuBarnesOzgur_TIT, wu2017geometry} and is the key step in resolving Cover's problem.

\begin{proposition}\label{P:SDPI}
For the Markov chain described in \eqref{e:lmc}, if $I(Z^n;U_n|Y^n)\leq nC_0$, then   $I(X^n;U_n|Y^n)$ is upper bounded by  
\begin{align}
I(X^n;U_n|Y^n)\leq \max_{C'\in [0,C_0]}\min_{r>0}\frac{n}{2}\ln  \frac{P(N+1-2 e^{-C'}\sqrt{N(1-e^{-2r})}  )+N (1-e^{-2C'}(1-e^{-2r}))}{(P+1)Ne^{-2r}} .   \label{eq:longeq_mi}
\end{align}
 \end{proposition}

Proposition \ref{P:SDPI} allows us to derive a new upper bound on the relay channel. In particular, if we use $U_n$ to denote the relay's transmission over the bit pipe, then it is easy to see that $Y^n-X^n-Z^n-U_n$ for the relay channel satisfies the conditions of the Markov chain described in \eqref{e:lmc}, and $$I(Z^n;U_n|Y^n)=H(U_n|Y^n)-H(U_n|Z^n,Y^n)\leq H(U_n)\leq nC_0.$$  Therefore,  by Fano's inequality and  Proposition \ref{P:SDPI} we can bound $C(C_0)$ by 
\begin{align} 
  nC(C_0) &  \leq I(X^n; Y^n, U_n) + n \epsilon \nonumber \\
& =  I(X^n; Y^n ) + I(X^n; U_n|Y^n )+ n \epsilon \nonumber \\
&\leq\max_{C'\in [0,C_0]}\min_{r>0} \frac{n}{2} \ln  \frac{P(N+1-2 e^{-C'}\sqrt{N(1-e^{-2r})}  )+N (1-e^{-2C'}(1-e^{-2r}))}{ Ne^{-2r}} + n\epsilon     \label{eq:longeq_capacitybound}
\end{align}
where we have used the simple fact that $I(X^n; Y^n )\leq\frac{n}{2}\ln(1+P)$. The upper bound in \eqref{eq:longeq_capacitybound} resolves Cover's problem as one can easily verify that it is strictly smaller than $n\,C(\infty)$ for any finite $C_0$.

\subsection{Proof of SDPI via Transportation Inequality}

To prove Proposition \ref{P:SDPI}, we need the following lemma, which is a consequence of the conditional transportation inequality stated in Theorem \ref{T:CNTI_Info}.  

\begin{lemma}\label{l:condlemma}
For the Markov chain \eqref{e:lmc}, if $I(Z^n; U_n|X^n)=nC'$ for some $C'\geq 0$, then for any $r>0$ there exists a random vector $\bar{Z}^n$ such that: 
\begin{enumerate}
\item  $P_{X^n,\bar{Z}^n,U_n}=P_{X^n,Z^n,U_n}$;
\item   $\dE[\bar{Z}^n \cdot Y^n] \ge n(P+\sqrt{N(1-e^{-2r})}e^{-C'})$; 
\item  $I(\bar Z^n;Y^n| X^n, U_n)\le nr.$ 
\end{enumerate}
\end{lemma}

\begin{proof}
Lemma \ref{l:condlemma} follows immediately from Theorem \ref{T:CNTI_Info} by setting $T=(X^n,U_n)$. In particular, 
noting that the random vector $W^n_2=Y^n-X^n\sim \mathcal N(0,I_n)$ and is independent of $(X^n,U_n)$, we have by Theorem \ref{T:CNTI_Info} that
\begin{align}
 W_2^2(P_{Z^n|X^n,U_n},P_{Y^n-X^n|X^n,U_n}|P_{X^n,U_n};nr)
 &\leq \dE[\|Z^n\|^2]+n-2n\sqrt{1-e^{-2r}}\sqrt{\frac{1}{2\pi e}e^{\frac{2}{n}h(Z^n|X^n,U_n)}}.   \nonumber
\end{align}
Therefore, there exists a random vector $\bar{Z}^n$ such that
\begin{align}
 (\bar{Z}^n ,X^n,U_n)& \sim P_{Z^n,X^n,U_n},\nonumber\\
I(\bar{Z}^n;Y^n|X^n,U_n)& \leq nr, \nonumber\\
 \dE[\bar Z^n\cdot (Y^n-X^n)] & \ge n\sqrt{1-e^{-2r}}\sqrt{\frac{1}{2\pi e}e^{\frac{2}{n}h(Z^n|X^n,U_n)}}. \label{Pf:NTI_Info_Cond_Lemma_Cond3}
\end{align} 
This proves 1) and 3) of Lemma \ref{l:condlemma}. To show 2) of Lemma \ref{l:condlemma}, note that
\begin{align}
\dE[\bar Z^n\cdot Y^n]&=\dE[\bar{Z}^n\cdot (Y^n-X^n)]+\dE[\bar Z^n\cdot X^n] \nonumber\\
& \ge  n\sqrt{1-e^{2r}}\sqrt{\frac{1}{2\pi e}e^{\frac{2}{n}h(Z^n|X^n,U_n)}}+\dE[Z^n\cdot X^n] \label{Pf:NTI_Info_Cond_Lemma_ex1}\\
&=n\sqrt{1-e^{2r}}\sqrt{Ne^{-2C'}}+nP\label{Pf:NTI_Info_Cond_Lemma_ex2}\\
&=n(P+\sqrt{N(1-e^{-2r})}e^{-C'}) \nonumber
\end{align}
where \eqref{Pf:NTI_Info_Cond_Lemma_ex1} follows from \eqref{Pf:NTI_Info_Cond_Lemma_Cond3} and the fact that $(\bar Z^n,X^n)\sim P_{Z^n,X^n}$, and \eqref{Pf:NTI_Info_Cond_Lemma_ex2} holds because $$h(Z^n|X^n,U_n)=h(Z^n|X^n)-I(Z^n;X^n|U_n)=\frac{n}{2}\ln 2\pi e Ne^{-2C'}.$$ This completes the proof of Lemma \ref{l:condlemma}. 
\end{proof}

We now use Lemma \ref{l:condlemma} to prove Proposition \ref{P:SDPI}.  Assuming that for the Markov chain \eqref{e:lmc}, 
$$I(Z^n; U_n|X^n)=nC'$$ for some $C'\geq 0$, we can create an auxiliary random vector $\bar{Z}^n$ coupled with $X^n, U_n, Y^n$ so as to satisfy the properties in Lemma \ref{l:condlemma}. Therefore, we have
\begin{align}
 &I(X^n;U_n|Y^n)\nonumber\\
 &= I( \bar Z^n;U_n|Y^n)+ I(X^n;U_n|Y^n, \bar Z^n) - I( \bar Z^n;U_n| Y^n,X^n) \nonumber\\
& = I( \bar Z^n;U_n|Y^n)+ h( U_n|Y^n, \bar Z^n) - h(  U_n| Y^n,X^n)  \nonumber\\
&\leq  I( \bar Z^n;U_n|Y^n)+ h( U_n| \bar Z^n) - h(  U_n|  X^n)\nonumber\\
&=  I( \bar Z^n;U_n|Y^n)- I( \bar Z^n; U_n| X^n) \label{e:markovaux}\\
& = h(\bar Z^n|Y^n)-h(\bar Z^n|Y^n, U_n) - I(\bar  Z^n; U_n| X^n)  \label{e:tobound}
\end{align}
where \eqref{e:markovaux} follows because $P_{X^n,\bar{Z}^n,U_n}=P_{X^n,Z^n,U_n}$ by 1) of Lemma \ref{l:condlemma} and thus $X^n-\bar Z^n-U_n$ forms a Markov chain.   In the following, we will bound  
the first two terms in \eqref{e:tobound} respectively. (Note that this bounding process precisely mirrors the packing argument used in the geometric proof of \cite{WuBarnesOzgur_TIT} and \cite{wu2017geometry}, and provides an interpretation of the packing argument in terms of auxiliary random variables.)


To bound the first term in \eqref{e:tobound}, we have for any $r>0$,
  \begin{align}
 &h(\bar Z^n|Y^n)  =   h\left(\bar Z^n-\frac{\dE[\bar Z^n\cdot Y^n]}{\dE[\|Y^n\|^2]}Y^n\Big|Y^n\right) \nonumber \\
& \leq  h\left(\bar Z^n-\frac{\dE[\bar Z^n\cdot Y^n]}{\dE[\|Y^n\|^2]}Y^n \right) \nonumber \\
& \leq  \frac{n}{2}\ln \frac{2\pi e}{n} \dE\left[\left\|\bar Z^n-\frac{\dE[\bar Z^n\cdot Y^n ]}{\dE[\|Y^n\|^2]}Y^n \right\| ^2\right]   \nonumber \\
& = \frac{n}{2}\ln  \frac{2\pi e}{n}\left(\dE[\|\bar Z^n\|^2]-\frac{\dE[\bar Z^n\cdot Y^n]^2}{\dE[\|Y^n\|^2]}\right)  \nonumber\\
& \leq \frac{n}{2}\ln 2\pi e \frac{P(N+1-2 e^{-C'}\sqrt{N(1-e^{-2r})}  )+N (1-e^{-2C'}(1-e^{-2r}))}{P+1}\label{e:comb2} 
 \end{align}
 where in the last step we have used 2) of Lemma \ref{l:condlemma}.  To bound the second term in \eqref{e:tobound}, we have for any $r>0$,
 \begin{align}
& h(\bar Z^n|Y^n, U_n)  \geq  h(\bar Z^n|Y^n, U_n,X^n) \nonumber \\
 & = h(\bar Z^n| U_n,X^n) -I(\bar Z^n; Y^n|U_n,X^n) \nonumber \\
  & = h( \bar Z^n|  X^n) - I( \bar Z^n; U_n| X^n) -I(\bar Z^n; Y^n|U_n,X^n) \nonumber \\
  &\geq \frac{n}{2}\ln 2\pi e N - nC' -nr \nonumber \\
  &= \frac{n}{2}\ln 2\pi   N e^{1-2(C'+r) }  \label{e:comb3}
 \end{align}
where the second inequality follows from 3) of Lemma \ref{l:condlemma}.

Plugging \dref{e:comb2}--\dref{e:comb3}  into \eqref{e:tobound} gives a bound on $I(X^n;U_n|Y^n)$ in terms of the value of $I(Z^n; U_n|X^n)=nC'$ that holds for any $r>0$. Therefore, the bound can be tightened by minimizing over $r>0$. The value $C'$ is unknown, but due to the Markov chain \eqref{e:lmc} we have 
$$I(  Z^n; U_n| X^n)\leq I(  Z^n; U_n| Y^n)\leq nC_0.$$
The bound in Proposition  \ref{P:SDPI} follows by taking a maximum over $C'\in [0, C_0]$.

\appendices

\section{Comparison of \eqref{E:TI} and \eqref{E:NewTI}}\label{A:Comp}
Given $P_{Y^n} = \mathcal N(0, I_n)$ and $P_{Z^n}\ll P_{Y^n}$,  let $f_{Y^n}$ and $f_{Z^n}$ denote their respective densities. Then we have
\begin{align*}
& \text{R.H.S. of  \eqref{E:TI}} =2\dE\left [\ln \frac{f_{Z^n}(Z^n)}{f_{Y^n}(Z^n)} \right] \\
& =2 \dE[\ln f_{Z^n}(Z^n) ]-2 \dE\left[\ln \frac{1}{(2\pi )^{n/2}} \exp\left(-\frac{\|Z^n\|^2}{2}\right)\right]  \\
& = -2h(Z^n)+n\ln 2\pi +\dE[\|Z^n\|^2] \\
& =\dE[\|Z^n\|^2] +n  - 2 n\left[ \frac{1}{2}\left(\frac{2}{n}h(Z^n)- \ln 2\pi e \right) +1\right]  \\
& =\dE[\|Z^n\|^2] +n  - 2 n\left[ \ln \sqrt{\frac{1}{2\pi e}  e^{\frac{2}{n}h(Z^n)}  }+1\right]  \\
& \geq \dE[\|Z^n\|^2] +n  - 2 n   \sqrt{\frac{1}{2\pi e}  e^{\frac{2}{n}h(Z^n)}  }   \\
&= \text{R.H.S. of  \eqref{E:NewTI}}
\end{align*}
where the inequality follows from $\ln a+1 \leq a$ and  holds with equality iff $  \sqrt{\frac{1}{2\pi e}  e^{\frac{2}{n}h(Z^n)}  }  =1$, i.e. $h(Z^n)=\frac{n}{2}\ln 2\pi e$. 

\section{Proofs of Proposition \ref{P:IDCTRI}}\label{A:TRI}
 
Let $(X_1,X_2,X_3)\sim P$ be a coupling of  $(\mu_1,\mu_2, \mu_3)$ such that 
\begin{enumerate}
\item $X_1=g(X_2)$ where  $g$ is a one-to-one mapping and  $(g(X_2), X_2)$ is an optimal coupling of $(\mu_1,\mu_2)$ that attains $W_p (\mu_1,\mu_2)$, i.e., 
$$W_p (\mu_1,\mu_2)=\left\{ \dE_{P}[d^p(X_1,X_2)]\right\}^{1/p};$$
\item $(X_2,X_3)$ is an optimal coupling of $(\mu_2,\mu_3)$ under  the $( R,\tau,\delta )$-information density constraint that attains $W_p(\mu_2,\mu_3; R,\tau,\delta )$, i.e., 
$$W_p(\mu_2,\mu_3; R,\tau,\delta )=\left\{ \dE_{ P}[d^p(X_2, X_3)]\right\}^{1/p}.$$
\end{enumerate}
From the above two conditions, it follows that $(X_1,X_3)$ is a coupling of $(\mu_1,\mu_3)$ that also satisfies the $( R,\tau,\delta )$-information density constraint. Indeed, since $X_1$ and $X_2$ are one-to-one mappings of each other, we have
\begin{align}
I_P(X_1;X_3)=I_P(X_2;X_3)\leq R \nonumber 
\end{align}
and 
\begin{align}
 \mathbb P  (|i_{P}(X_1;X_3)- I_P(X_1;X_3)| \leq \tau) 
= \mathbb  P (|i_{P}(X_2;X_3)- I_P (X_2;X_3)| \leq \tau) >1-\delta .\nonumber 
\end{align}
Therefore, we have 
\begin{align}
 W_p(\mu_1,\mu_3; R,\tau,\delta )& \leq \dE_P [d(X_1,X_3)^p]^{1/p} \nonumber \\
& \leq  \dE_P[(d(X_1,X_2)+d(X_2,X_3))^p]^{1/p}\nonumber\\
& \leq  \dE_P[d(X_1,X_2)^p]^{1/p}+\dE_P[d(X_2,X_3)^p]^{1/p} \label{Pf:Wass_InfoD_Tr_Ex1}\\
& = W_p(\mu_1,\mu_2)+W_p(\mu_2,\mu_3; R,\tau,\delta )\nonumber 
\end{align}
where \eqref{Pf:Wass_InfoD_Tr_Ex1} follows from the Minkowski inequality. This completes the proof of Proposition \ref{P:IDCTRI}.

\section{Continuity of $\alpha(\eta)$}

\begin{lemma}\label{l:W2alphaisCont}
	The function $\alpha(\eta)$ defined in \eqref{D:W_InfD_alpha} is continuous in $\eta$. 
\end{lemma}
\begin{proof}
Rewrite $\alpha(\eta)$ as 
\begin{align*}
\alpha(\eta) &=\ \frac{1}{n}\left(-\sup_{x^n\in B}\{\ln \mu(\text{Cap}(x^n,\eta)\cap A)\}+\ln (\mu (A))-R\right)\\
  &=\ \frac{1}{n}\left(-\ln\left( \sup_{x^n\in B}\{ \mu(\text{Cap}(x^n,\eta)\cap A)\}\right)+\ln (\mu (A))-R\right).
\end{align*}
To prove $\alpha(\eta)$ is continuous  in $\eta$, it suffices to show  
\begin{align}
\sup_{x^n\in B}\mu(\text{Cap}(x^n,\eta)\cap A) \label{e:rightconti}
\end{align}
 is continuous in $\eta$.  

For any $\epsilon>0$, we have
 \begin{align}
&\left|\sup_{x^n\in B}\mu(\text{Cap}(x^n,\eta+\epsilon)\cap A)-\sup_{x^n\in B}\mu(\text{Cap}(x^n,\eta)\cap A)\right| \nonumber \\
= \ & \sup_{x^n\in B}\mu(\text{Cap}(x^n,\eta+\epsilon)\cap A)-\sup_{x^n\in B}\mu(\text{Cap}(x^n,\eta)\cap A) \nonumber \\
= \ & \sup_{x^n\in B} \left\{\mu(\text{Cap}(x^n,\eta+\epsilon)\cap A)-\sup_{x^n\in B}\mu(\text{Cap}(x^n,\eta)\cap A) \right\}\nonumber \\
\leq \ &    \sup_{x^n\in B} \bigg\{\mu(\text{Cap}(x^n,\eta+\epsilon)\cap A)- \mu(\text{Cap}(x^n,\eta)\cap A) \bigg\}\nonumber \\
= \ &  \sup_{x^n\in B}  \mu((\text{Cap}(x^n,\eta+\epsilon) \setminus  \text{Cap}(x^n,\eta))\cap A)  \nonumber \\
\leq \ &  \sup_{x^n\in B} \mu( \text{Cap}(x^n,\eta+\epsilon) \setminus  \text{Cap}(x^n,\eta) )  \nonumber \\
\leq \ & \delta (\epsilon)
 \end{align}
for some $\delta (\epsilon)\to 0$ as $\epsilon \to 0$, and therefore we have shown the right-continuity of \dref{e:rightconti}. Similarly, we can show the left-continuity of \dref{e:rightconti}. This proves the lemma.
\end{proof}

\section*{Acknowledgement}
The authors are grateful to Mokshay Madiman and Igal Sason for their helpful discussions and comments.


\end{document}